\documentclass[twocolumn]{article}

\usepackage[bindingoffset=0cm,margin=2.2cm]{geometry}

\usepackage[pdftex]{graphicx}
\usepackage{sidecap}
\usepackage{url}

\usepackage[colorlinks]{hyperref}
\usepackage[all]{hypcap}

\usepackage{amsthm}
\usepackage{amssymb}
\usepackage{amsmath}
\usepackage{algorithm2e}
\usepackage{marginnote}

\usepackage{multirow}
\usepackage{xfrac}

\usepackage[font={small,sf}]{caption}

\usepackage{setspace} 
\newtheorem{theorem}{Theorem}[section]

\newtheorem{proposition}[theorem]{Proposition}
\newtheorem{corollary}[theorem]{Corollary}

\newcommand{\bb}[1]{\pmb{\mathrm{#1}}}

\def\Nn{\mathcal{N}}
\def\Ss{\mathcal{S}}

\def\RR{\mathbb{R}}
\def\Sym{\mathrm{Sym}\,}
\def\vol{\mathrm{vol}}
\def\Vol{\mathrm{Vol}}
\def\Area{\mathrm{Area}}
\def\dis{\mathrm{dis}}
\def\TV{\mathrm{V_{\Ss}}}
\def\OO{\mathrm{O}}
\def\CC{\mathrm{C}}

\begin{document}

\title{Probably Approximately Symmetric: \\Fast Rigid Symmetry Detection with Global Guarantees}

\author{Simon Korman, Roee Litman, Shai Avidan and Alex Bronstein
        \\
        School of Electrical Engineering, Tel Aviv University, Israel
       }
\date{}

\maketitle

\begin{figure*}
\begin{center}
\vspace{-18pt}
    \addtolength{\tabcolsep}{12pt}
    \begin{tabular}{ccc}
        \includegraphics[width = 0.23\textwidth]{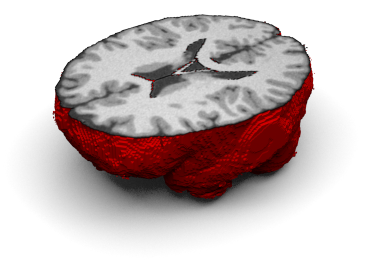} &
        \includegraphics[width = 0.15\textwidth]{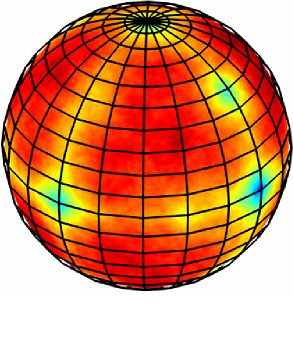} &
        \includegraphics[width = 0.26\textwidth]{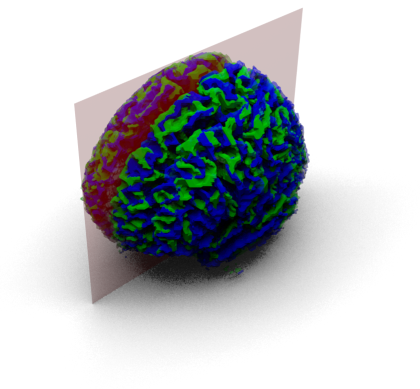}
    \end{tabular}
    \caption{\label{fig.MRI}\textbf{Symmetry in an MRI scalar volume.}
    \textbf{Left:} an illustration of the bottom half of a volumetric MRI of a brain. Our method is applicable to any volumetric image, not necessarily a binary one representing a solid 3D shape.
    \textbf{Center}: visualization of the reflection distortion (symmetry error) of the volume. Position on the sphere represents the normal direction of the reflection plane through the volume center, and color code represents the distortion (increasing from blue to red).
    \textbf{Right:} The detected reflection symmetry plane, along with the original image (green) and its reflected version (blue), visualized as iso-surfaces.}
    \vspace{12pt}
\end{center}
\end{figure*}

\begin{abstract}
We present a fast algorithm for global rigid symmetry detection with approximation guarantees. The algorithm is guaranteed to find the best approximate symmetry of a given shape, to within a user-specified threshold, with very high probability.
Our method uses a carefully designed sampling of the transformation space, where each transformation is efficiently evaluated using a sub-linear algorithm.
We prove that the density of the sampling depends on the total variation of the shape, allowing us to derive formal bounds on the algorithm's complexity and approximation quality.
We further investigate different volumetric shape representations (in the form of truncated distance transforms), and in such a way control the total variation of the shape and hence the sampling density and the runtime of the algorithm.
A comprehensive set of experiments assesses the proposed method, including an evaluation on the eight categories of the COSEG data-set. This is the first large-scale evaluation of any symmetry detection technique that we are aware of.
\end{abstract}

\section{Introduction}

Symmetry is ``a distinction without a difference'' in the words of the renowned physicist and Nobel laureate Frank Wilczek. Doubtlessly, symmetry along with the related
concepts of self-similarity and invariance is an all-pervasive property of Nature and man-made art.
Engineering, architectural, and artistic designs characterized by symmetry usually enjoy structural robustness and efficiency, which also explains why evolution led  many biological constructions to assume symmetric properties. Symmetry also plays an important role in our visual perception, in particular that of beauty, and according to modern physical theories, is incorporated deeply into the laws of the universe itself.

Automatic detection of symmetries of a 3D geometric shape has received significant attention in computational geometry, computer graphics, and vision literature. However, despite the steady progress in the field, the task remains computationally challenging.
Existing approaches to symmetry detection interpret symmetry as invariance under a certain class of transformations and they can be categorized according to several key features.
The taxonomy we present here is by no means complete, and the reader is referred to \cite{mpwc_symmSurvey_12,hel2010computational} for a comprehensive survey of symmetry detection in 3D shapes and images.

First and foremost, symmetry is characterized by a group of admissible transformations. While it is customary to tacitly assume the Euclidean group (defining \emph{rigid} symmetries or \emph{congruences} further categorized into reflections or involutions, rotations, and improper rotations or roto-reflections including the former two), more elaborate types of transformations involving uniform scaling (\emph{similarities}), affine and projective transformations, and even \emph{intrinsic} symmetries have been studied in the literature. Our main focus will be restricted to rigid symmetries, though the proposed algorithm and analysis can be extended to practically any group with a finite (and reasonably small) number of parameters, such as the affine group.

Second, symmetries can be classified as \emph{global}, \emph{partial}, and \emph{local}. Global symmetry is defined by a transformation that maps the whole shape onto itself. A shape not possessing a global symmetry can still have partial symmetries in the form of self-similar parts. Local symmetry usually refers to regular spatial arrangements of a structural element into tilings and ornaments. Here, we focus on global symmetries.

Finally, \emph{exact} (perfect) and \emph{approximate} (imperfect) symmetries can be distinguished: the former map the shape exactly to itself, whereas in case of the latter the mapping leads to a distortion smaller than a pre-defined threshold. Depending on the application, exact partial symmetries can also be regarded as approximate global ones. This work focuses on approximate symmetries.

\subsection{Prior work}
\label{sec:prior-work}

\paragraph*{Exact symmetry.}
Efficient algorithms exist for exact symmetry detection. For example, in the case of a collection of $n$ points in the plane, Atallah \cite{Atallah85} describes an algorithm for enumerating all axes of symmetry under reflection of a planar shape. Wolter {\em et al.} \cite{Wolter85} give exact algorithms, based on string matching, for the detection of symmetries of point clouds, polygons, and polyhedra. 
These algorithms are often impractical due to their sensitivity to noise, because they are restricted to exact symmetries.

\paragraph*{Non-parametric symmetry.}
Several algorithms exist for the detection of non-parametric intrinsic symmetries of deformable shapes (as opposed to the rigid \emph{extrinsic} counterparts).
Raviv {\em et al.} \cite{Raviv} use a branch-and-bound technique to find global intrinsic symmetries with a prescribed distortion of pairwise geodesic distances.
Ovsjanikov {\em et al.} \cite{Ovsjanikov:2008:GIS} detects the global {\em intrinsic} symmetry of shapes using the spectral properties of the Laplace-Beltrami operator.
Lipman {\em et al.} \cite{LipmanCDF10} relieve the assumption of a known transformation group by introducing a symmetry-factored embedding, which enables detecting approximate, as well as partial symmetries of a point cloud, also using spectral methods.

\paragraph*{Approximate symmetry.}
The detection of approximate symmetries has also been addressed in the literature, and can be roughly divided into two approaches:

The first approach defines approximate symmetry by an infimum of a continuous distance function quantifying how similar is a shape to its transformed version.
Zabrodsky {\em et al.} \cite{zabrodsky1995symmetry} proposed such a symmetry distance,
which has been largely adopted and extended by following works, including our approach.
One way to detect an approximate infimum is through an exhaustive evaluation of the transformation space on a grid with a high-enough density.
This task can be done na\"{i}vely in $\mathcal{O}(n^6)$ for a shape discretized by an $n \times n \times n$ grid, and reaching high accuracy in this approach requires bigger $n$ and an increase in computation times.
A more efficient algorithm by Kazhdan {\em et al.} \cite{kazhdan2004symmetry} performs the task in $\mathcal{O}(n^4)$ using an FFT-like approach, but does not provide guarantees on the distance from the optimal possible distortion (See Section~\ref{sec:kazhdan-comparison} for a more detailed comparison to our method).

The second approach alleviates the computational complexity by translating the search into a proxy domain, realizing that the set of admissible symmetries is sparse in the transformation space.
One of the earliest examples is \cite{Sun1997}, which uses the gaussian image as the proxy domain.
Later work by Martinet {\em et al.} \cite{martinet2006accurate} examines extrema and spherical harmonic coefficients of generalized even moments.
Mitra {\em et al.} \cite{MitraGP06} cluster Hough-like votes for transformations that align boundaries with similar local shape descriptors.
Podolak {\em et al.} \cite{podolak2006planar} detect reflection symmetries using a monte-carlo algorithm that selects a pair of surface points and votes for the plane between them.
Searching a proxy space provides a set of candidate transformations, that have to be validated directly using some symmetry measure (e.g., \cite{zabrodsky1995symmetry}). These candidate transformation are usually further refined using e.g. \emph{Iterative Closest Point} (ICP).
Consequently, there is no guarantee on how far is the symmetry measure of the detected symmetries from that of the optimal one.

Our algorithm follows the first approach in that it samples the transformation space, but does so in an efficient manner that guarantees a known approximation error. This enables the use of a branch-and-bound scheme that allows to match the performance of the second approach while maintaining approximation guarantees.

\paragraph*{Alignment methods}
A method related to ours is the surface registration algorithm of~\cite{aiger20084} that uses a clever sampling of transformation space.
However, it is not easily applicable to the problem of detecting all approximate symmetries of a shape.

Our work follows that of \cite{FastMatchPaper}, who proposed a fast method for 2D affine template matching in images with global guarantees. As opposed to their sampling density which depends on a generic image assumption (e.g., image smoothness), ours is determined adaptively according to the specific shape `complexity'. Additionally, we manipulate the shape representation to control the sampling density, and hence the algorithm's runtime. Finally, while template matching focuses on finding a single best transformation, our goal is to detect all such transformations.

\subsection{Contributions}
We detect global rigid symmetries in volumetric representations of 3D shapes, and introduce an algorithm that is guaranteed, with high probability, to detect the best symmetry within a given degree of approximation.
This is inspired by the classical ``probably approximately correct'' (PAC) framework~\cite{valiant1984theory} in learning theory (hence the title of the present paper).
To the best of our knowledge, this is the first symmetry detection algorithm coming with such a guarantee.
An example output of the algorithm, detecting the bilateral symmetry in a brain MRI image, can be seen in Figure \ref{fig.MRI}.

We provide a bound on the required sampling density of transformation space, which is the basis of our algorithm. This bound depends on the desired approximation level as well as, surprisingly, on the `complexity' of the specific shape, which is manifested through the total variation of its volumetric representation.
We further show how to construct shape representations with reduced total variation leading to reduced complexity, and discuss the tradeoff between complexity and sensitivity to noise.

A comprehensive experimental evaluation validates that our approach is capable of detecting approximate symmetries in a large data-set, as well as detecting all symmetries in complicated shapes, all within state-of-the-art execution times.

\section{Approximate rigid symmetries}

We start by defining our level-set based shape representation and approximate symmetry. We then bound the sample density required to detect approximate symmetries with a user specified precision parameter $\delta$.
Generally speaking, although we keep all derivations in the continuous setting, they are straightforwardly amenable to any reasonable discretization of the volume, including hierarchical subdivisions.

\subsection{Shape representation}
\label{seq:representation}

Let $\Ss$ be a three-dimensional rigid shape with the centroid aligned at the origin.
We represent the shape by the $\frac{1}{2}$-sub-level set of a level set function $s : \RR^3 \mapsto [0,1]$.
The simplest of such representations is the binary indicator function of $\Ss$ (equalling 1 in the interior); other representations such as truncated distance maps will be discussed in Section~\ref{sec.shape.complex}. We will freely interchange between $\Ss$ and $s$ referring to a shape.
We will take the \emph{radius} of the shape to be the smallest scalar $r$ such that the function $s$ is invariant to rotations outside the Euclidean ball $B_r(0)$ of radius $r$ centered at the origin,
\begin{eqnarray}
r &=& \inf_{r} \{ s(\bb{R}\bb{x}) = s(\bb{x}) : \bb{x} \in {B_r(0), \bb{R} \in \mathrm{SO}(3) }  \}.
\end{eqnarray}
The ball $B_r(0)$ defines the effective support of $s$, which might be larger than the shape $\Ss$.

We associate with the shape the \emph{total variation} of $s$,
\begin{eqnarray}
\TV &=& \frac{1}{\Vol\,B_r} \int_{B_r} \| \nabla s(\bb{x}) \| d\vol(\bb{x}),
\end{eqnarray}
where $d\vol$ denotes the standard volume element. When $s$ is not differentiable, total variation can be defined using the weak derivative.
In particular, for the case of the indicator function $\TV$ is equal to the ratio between the area of the boundary $\partial S$ and the volume of the bounding ball $B_r$.
Geometrically, total variation can be related to the total curvature of the shape and the amount of ``features'' it contains.
Note that for the case of $\OO(3)$, one could have considered derivation and integration only tangent to concentric spheres.

\subsection{Rigid symmetries}

Let $T \in \mathrm{E}(3)$ be a Euclidean transformation (a combination of translation, rotation, and reflection).
The transformed shape $T\Ss$ will be represented by the indicator function $s (T \bb{x})$.
$T$ is said to be an \emph{exact global symmetry} of $\Ss$ if $s(T\bb{x}) = s(\bb{x})$.
The collection of all symmetries of $\Ss$ forms a group under function composition, which we refer to as the \emph{symmetry group} of $\Ss$, denoted by $\Sym\,\Ss$.
Each symmetry $T \in \Sym\,\Ss$ defines a collection of stationary points, $\{ \bb{x} = T\bb{x} \}$, which is known to be either a line or a plane.
Such a line or plane is called a \emph{symmetry axis} (or \emph{plane}) of the shape.
The set of symmetry axes and planes fully defines the symmetry group of a shape.
Since translations have no stationary points, for compactly supported shapes,
$\Sym\,\Ss$ is necessarily a subgroup of the orthogonal group $\OO(3)$ containing rotations and reflections around the shape's centroid.
For this reason, we will henceforth denote the admissible transformations as $3 \times 3$ rotation or reflection matrices, $\bb{R}$.

Exact symmetries are a mathematical idealization rarely achieved in practice due to acquisition and representation inaccuracies.
In order to account for such imperfections, we define the distortion
\begin{eqnarray}
\dis_s\, \bb{R} &=& \frac{1}{\Vol\, B_r} \| s - \bb{R}^{-1}(s) \|_1 \label{def:distortion}\\
&=& \frac{1}{\Vol\,B_r} \int_{B_r} | s(\bb{x}) - s(\bb{R}\bb{x}) | d\vol(\bb{x}), \nonumber
\end{eqnarray}
where $\bb{R}^{-1}(s)$ is a short-hand notation for $(s\circ \bb{R})(x)=s(\bb{R}\bb{x})$. Note that $\dis_s\, \bb{R}$ is bounded to the interval [$0,1$], and equals zero in the case of perfect symmetry.
For solid shapes represented by indicator functions, the distortion can be interpreted as the total amount of mismatched volume and it coincides with the common symmetry measure of \cite{zabrodsky1995symmetry}. Note that such an $L_1$ formulation is more robust to outliers compared to, e.g., the worst-case Hausdorff distance.

We say that $\bb{R} \in \OO(3)$ defines an $\epsilon$-symmetry of $\Ss$ if $\dis_s\,\bb{R} \le \epsilon$,
and denote by $\Sym_\epsilon \Ss$ the collection of all $\epsilon$-symmetries of $\Ss$. Note that unlike their exact counterparts, approximate symmetries do not necessarily form a group.

%

\subsection{Sampling of the orthogonal group}

In order to practically detect symmetries, one necessarily has to work with a finite sample of the transformation space (i.e. the orthogonal group). The main ingredient of our approach is an upper bound on the sampling density controlled by the maximum allowed distortion of an approximate symmetry.

We begin by defining a metric between any two transformations in the space, which will be used later to define a net of transformations. The metric measures how far apart any point in the ball $B_r$ may be mapped by two different transformations, formally:
\begin{eqnarray}
D (\bb{R}_1,\bb{R}_2) = \max_{\|\bb{x}\| \le r} \| \bb{R}_1 \bb{x} - \bb{R}_2 \bb{x} \|.
\end{eqnarray}
Note that this distance  does not depend on the shape, but rather only on its support radius $r$.

A key observation is that the difference in the distortion of two transformations is upper bounded by the product of the shape total-variation $\TV$ and the distance $D$ between the transformations. This is formalized in the following proposition, with the accompanying illustration in Figure~\ref{fig.example}.
\begin{proposition}\label{prop:bound}
$|\dis_s\,\bb{R}_1 - \dis_s\,\bb{R}_2| \le \TV \cdot D(\bb{R}_1,\bb{R}_2)$ for any $\bb{R}_1, \bb{R}_2 \in \OO(3)$.
\end{proposition}

\begin{proof}
First, observe that invoking the triangle inequality and using the group properties,
\begin{eqnarray}
|\dis_s\,\bb{R}_1 &-& \dis_s\, \bb{R}_2| \quad \ldots \nonumber \\
&=& \frac{1}{\Vol\,B_r}\cdot| \|s - \bb{R}_1^{-1}(s) \|_1 - \| s - \bb{R}_2^{-1}(s) \|_1 | \nonumber \\
 &\le& \frac{1}{\Vol\,B_r}\cdot\| \bb{R}_1^{-1}(s) -  \bb{R}_2^{-1}(s) \|_1\nonumber \\
  &=& \frac{1}{\Vol\,B_r}\cdot\| q - \bb{R}_1^{-1}(q) \|_1 = \dis_q\,\bb{R}\label{eq:dis} 
\end{eqnarray}
with $q = \bb{R}^{-1}(s)$ and $\bb{R} = \bb{R}_2 \bb{R}_1^{-1}$.
We can therefore define a new shape $\mathcal{Q} = \bb{R}_1 \Ss$ with the corresponding function $q$, and operate
with $\dis_s\,\bb{R}$.
Using the group properties, it is also straightforward that $D(\bb{R}_1,\bb{R}_2) = D(\bb{I},\bb{R})$, with $\bb{I}$ being the identity transformation.

$\quad$ We define the flow $\Phi_{\bb{R}} : (\bb{x},t) \rightarrow \bb{R}^t \bb{x}$, $t \in [0,1]$,
inducing the orbits $C(\bb{x}) = \{ \bb{R}^t \bb{x} : \bb{x} \in B_r, t \in [0,1] \}$, whose length is upper-bounded by $D(\bb{I},\bb{R})$ .
Using the triangle inequality,
\begin{eqnarray}
| q(\bb{x}) - q(\bb{R}\bb{x}) |  &=&  | q \circ \Phi_{\bb{R}}(\bb{x},0) - q \circ \Phi_{\bb{R}}(\bb{x},1) |\label{eq:orbit}\\
& \le & \int_0^1 \| \nabla (q \circ \Phi_{\bb{R}}(\bb{x},t)) \| \, \| \dot{\Phi}_{\bb{R}} (\bb{x},t) \| dt \nonumber
\end{eqnarray}
where $\dot{\Phi}_{\bb{R}}(\bb{x},t) = \frac{\partial}{\partial t} \Phi_{\bb{R}}(\bb{x},t)$.
We can now derive that
\begin{eqnarray}\label{eq:integrals}
\lefteqn{\Vol\,B_r \cdot \dis_q\,\bb{R} = \int_{B_r} |q(\bb{x}) - q(\bb{R}\bb{x})| d\vol(\bb{x})}\nonumber\\
&& \le  \int_{B_r} \int_0^1  \| \nabla (q \circ \Phi_{\bb{R}}(\bb{x},t)) \| \, \| \dot{\Phi}_{\bb{R}} (\bb{x},t) \| dt \, d\vol(\bb{x})  \nonumber\\
&& =  \int_0^1  \int_{B_r}  \| \nabla (q \circ \Phi_{\bb{R}}(\bb{x},t)) \| \, \| \dot{\Phi}_{\bb{R}} (\bb{x},t) \|  d\vol(\Phi_{\bb{R}}(\bb{x},t)) \, dt   \nonumber\\
&& \le  D(\bb{I},\bb{R}) \cdot \int_{B_r}  \| \nabla (q \circ \Phi_{\bb{R}}(\bb{x},t)) \|    d\vol(\Phi_{\bb{R}}(\bb{x},t))  \nonumber\\
&& = D(\bb{I},\bb{R}) \cdot \int_{B_r}  \| \nabla q \|  d\vol(\bb{x})  = D(\bb{I},\bb{R}) \cdot  \Vol\,B_r \TV \;\; ,
\end{eqnarray}
where the first inequality follows from~(\ref{eq:orbit}), the following equality follows from the volume preservation under rotations and reflections, and the final equality from the definition of the total variation $\TV$.

$\quad$ Finally, the proposition follows by combining inequalities (\ref{eq:dis}) and (\ref{eq:integrals}) and using the fact that $D(\bb{R}_1,\bb{R}_2) = D(\bb{I},\bb{R})$.
\end{proof}


\begin{figure}[t]
\begin{center}
    \includegraphics[width=0.9\linewidth]{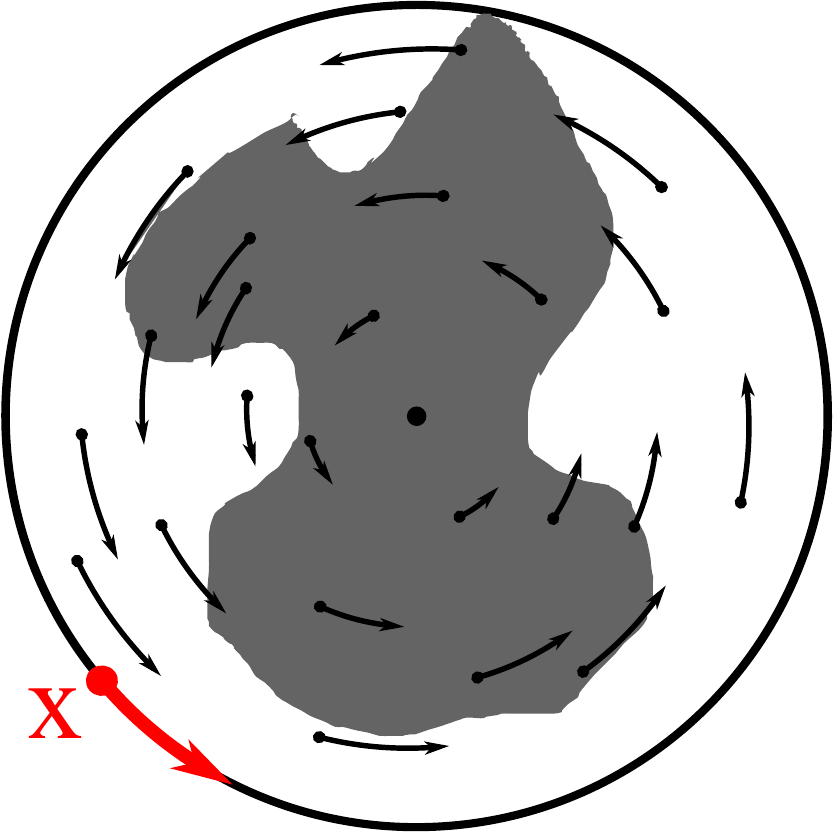}
    \caption{\textbf{A 2D illustration of Proposition~\ref{prop:bound}. } A planar shape (gray region) represented by a function $s$ undergoes a rotation $\bb{R}$. The change in value at every point is bounded by the accumulation of $\|\nabla s\|$ along the orbit it travels (marked by black arrows).
    The longest path $D(\bb{I},\bb{R}$) is traveled by the points farthest from the rotation axis, e.g. the one marked in red.
    Integrating these changes over the entire ball gives $dis_s \bb{R}$, which is bounded by $\TV \cdot D(\bb{I},\bb{R})$.}
    \label{fig.example}
\end{center}
\vspace{-25pt}
\end{figure}

Proposition~\ref{prop:bound} bounds the change in the distortion under a bounded displacement in the transformation space (the orthogonal group $\OO(3)$). We can now turn to defining a finite sampling of $\OO(3)$.
Let us fix a precision parameter $\delta>0$, set a sampling radius $\rho = \delta/\TV$, and construct a discrete set of transformations, $\Nn_\rho$, forming a $\rho$-net in $\OO(3)$ with respect to the distance $D$, namely that any point in $\OO(3)$ has a sample in $\Nn_\rho$ at a distance of at most $\rho$.

The rationale behind defining such a net of transformations is as follows: Let $\bb{R}^\ast \in \OO(3)$ be an $\epsilon$-symmetry of $\Ss$. While $\bb{R}^\ast$ will not necessarily be contained in the net $\Nn_\rho$, there will exist some other $\bb{R} \in \Nn_\rho$ with $\dis_s\,\bb{R} \le \epsilon+\delta$.\, In other words, evaluating the distortion of all the transformations in such a net guarantees the detection of symmetries within a predefined distortion.
In the following proposition, we describe an efficient construction of such a net and give a bound on its size.

\begin{proposition}\label{prop:sample-size}
Let $\Nn_\rho$ be a $\rho$-net in $\OO(3)$ with respect to the distance $D$. Then, $|\Nn_\rho| \le n_\rho = 4\pi\left(\frac{\rho}{r} - \sin \frac{\rho}{r}\right) \sim \mathcal{O}\left((\frac{r}{\rho})^3\right)$.
\end{proposition}

\begin{proof} First, from $D \le r \cdot d$ ($d$ being the standard geodesic distance on $\OO(3)$), we conclude that a $\rho'$-net (for $\rho' = \rho/r$) is a $\rho$-net in the metric $D$. We can therefore use the more convenient $d$, proceeding with the standard packing number argument:
The total volume of $\OO(3)$ is given by twice the area of the three-dimensional hypersphere, $2\Vol\,\mathbb{S}^3$.
Since $\Nn_\rho$ is $\rho'$-separated, the balls $B_{\rho'/2}$ in $\mathbb{S}^3$ form a disjoint collection, whose volume is
smaller than the total volume of $\OO(3)$. The bound is obtained by demanding
$n_\rho \Vol\,B_{\rho'/2} = 2\Vol\,\mathbb{S}^3$. Substituting closed form expressions for the volumes on the sphere yields
$n_\rho \pi (\rho' - \sin \rho')  = 4 \pi^2$,
from where $n_\rho$ is obtained.
Finally, using the Taylor expansion $\sin \rho' = \rho' - \frac{\rho^{\prime 3}}{6} + \mathcal{O}(\rho^{\prime 5})$ yields $\rho' - \sin \rho' \sim \mathcal{O}(\rho^{\prime 3})$.
\end{proof}

Our sampling of the orthogonal group is summarized in the following corollary, which is a direct consequence of the combination of Propositions~\ref{prop:bound} and \ref{prop:sample-size}, with the choice of $\rho = \delta/\TV$.

\begin{corollary}\label{coro:sampling}
For a given precision parameter $\delta>0$, there exists a sampling of the orthogonal group of size $\mathcal{O}( (\frac{r \cdot \TV }{\delta})^3 )$, such that for any given transformation $T$ in $\OO(3)$, the sample contains a transformation, whose distortion is bounded away by $\delta$ from the distortion of $T$.
\end{corollary}

Such a sampling of $\OO(3)$ can be achieved, using e.g. farther-point sampling \cite{eldar1997farthest} or fixed-step strategies. The choice of $\rho = \delta/\TV$ translates into an angular density of about $\frac{2\pi\cdot\delta}{r \cdot \TV }$ in the standard axis-angle parametrization.

\subsection{Fast evaluation of the distortion}\label{sec.fast-evaluation}

A na\"{\i}ve symmetry detection algorithm consists of testing whether $\dis_s\,\bb{R} \le \delta$ for each $\bb{R}$ in the net $\Nn_\rho$.
However, such a test requires the computation of the integral (\ref{def:distortion}) for each sample, which results in a non-trivial complexity.
To alleviate this burden, we use a faster randomized sub-linear sampling procedure, which gives approximately the same result with overwhelmingly high probability.

Let $\bb{x}_1,\cdots,\bb{x}_m$ be points randomly drawn from the uniform distribution on $B_r$.
We define the approximate distortion as
\begin{eqnarray}
\widetilde{\dis}_s \bb{R} = \frac{1}{m} \sum_{i=1}^m | s(\bb{x}_i) - s(\bb{R}\bb{x}_i)|,
\end{eqnarray}
where each of the summands is bounded on $[0,1]$.
Since $\mathbb{E}\{ \widetilde{\dis}_s\bb{R} \} = \dis_s\,\bb{R}$, we can use the Chernoff-Hoeffding inequality to bound the probability $P( |\widetilde{\dis}_s\bb{R} - \dis_s\,\bb{R}| > \epsilon)$, leading to the following
\begin{proposition}\label{prop:proptest}
For $m_\epsilon = \frac{1}{2 \epsilon^2} \log \frac{2}{p} = \mathcal{O}(\epsilon^{-2}\log\frac{1}{p})$,
$|\dis_s\,\bb{R} - \widetilde{\dis}_s \bb{R} | \le \epsilon$ with probability higher than $1-p$.
\end{proposition}

\section{Symmetry detection algorithm}

\begin{algorithm}[t]
\hrulefill
\SetAlgoLined
\SetKwInOut{Input}{input}\SetKwInOut{Output}{output}

\Input{Shape $\Ss$ represented by $s$; precision parameter $\delta>0$; error probability $p$ }

\Output{Approximate symmetry $\bb{R} \in \OO(3)$; approximate distortion $d$ }

\medskip

Construct a $\rho=\frac{\delta}{2\TV}$-net $\Nn_\rho$ on $\OO(3)$

\smallskip

\ForEach{$\bb{R} \in \Nn_\rho$}
{

\smallskip

Sample $m_{\delta/2} =\frac{2}{\delta^2} \log \frac{2}{p}$ random points from $B_r$
\smallskip

Compute $\widetilde{\dis}_s \bb{R} = \frac{1}{m} \sum_{i=1}^{m} | s(\bb{x}_i) - s(\bb{R}\bb{x}_i)|$
\smallskip
}

\smallskip

Return $\bb{R}$ with the minimal $d = \widetilde{\dis}_s \bb{R}$ \\
\hrulefill

\medskip

\caption{Best approximate symmetry detection.  \label{algo:single-symm} }

\end{algorithm}

Putting the pieces together, we summarize in Algorithm~\ref{algo:single-symm} the proposed method for detecting the best approximate symmetry.
Combining the previous results, we state the following
\begin{theorem}\label{thm.algorithm}
The runtime complexity of Algorithm~\ref{algo:single-symm} is $\mathcal{O}((r\cdot\TV)^3 \delta^{-5} \log \frac{1}{p})$ and with probability $1-p$, it holds that:
\begin{enumerate}
\item if $d\le 0.5\cdot\delta$, then $\bb{R}$ is a $\delta$-symmetry of $\Ss$
\item if $d > 1.5\cdot\delta$, then $\Ss$ has no $\delta$-symmetries.
\end{enumerate}
\end{theorem}
Observe that unless some elements of $\Nn_\rho$ are removed, the second condition will never happen, as the algorithm will return a symmetry $\delta$-close to the identity transformation.

\begin{figure}[b!]
\vspace{-3pt}
  \begin{center}
 \addtolength{\tabcolsep}{-1pt}
    \begin{tabular}{c c c  }
      \includegraphics[width = 0.14\textwidth]{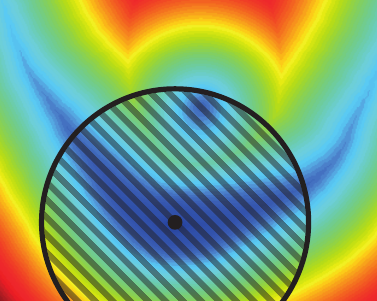} &
      \includegraphics[width = 0.14\textwidth]{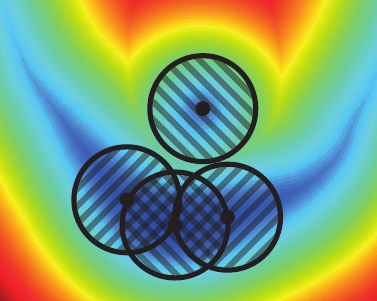} &
      \includegraphics[width = 0.14\textwidth]{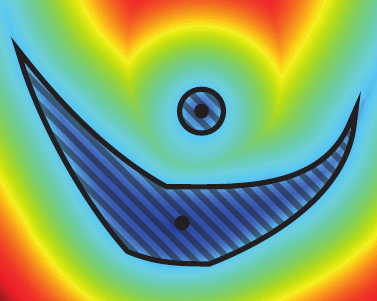} \\ 
      \small{too large radius} & \small{too small radius} & \small{flood-fill}\\
         \end{tabular}
         \vspace{-8pt}
\caption{\textbf{Illustration of minima neighbourhood removal.}
A 2D function, visualized using a heat-map.
In our implementation, we use a fixed-sized removal radius.
Using a too large radius may remove other minima, while a too small one leaves areas that might be detected in the next iteration.
These phenomena may be avoided
by applying a flood-fill procedure.
}
\label{fig:flood}
  \end{center}\vspace{-20pt}
\end{figure}

Algorithm~\ref{algo:single-symm} detects a single approximate symmetry of $\Ss$. In order to detect the entire $\Sym_\delta \Ss$, we run the algorithm sequentially, each time removing a neighborhood of the detected transformation $\bb{R}$ from $\Nn_\rho$.
The neighborhood can be naturally defined as the $\delta$-component of $\bb{R}$, computed by applying a flood-fill procedure to $\Nn_\rho$.
Alternatively, the neighborhood can be defined as a ball of a fixed radius with respect to the standard geodesic distance on $\OO(3)$. The latter approach was adopted in our experiments due to its simplicity, despite the problems that may arise when using a too small or too big radius (see Figure~\ref{fig:flood} for an illustration). 

Algorithm~\ref{algo:all-symm} summarizes the described procedure for the detection of all approximate symmetries.
When a rotation symmetry is detected, we further investigate its axis to find its $n$-fold symmetries (up to some integer $N$). We report on an $n$-fold rotation if the distortion of all its $n$ members is below $\delta$.
A continuous (axial) rotational symmetry is reported when the distortions of all members of all $n$-fold rotations, $n=2,\ldots,N$, are below $\delta$.

\begin{algorithm}[t]
\hrulefill
\SetKwInOut{Input}{input}\SetKwInOut{Output}{output}

\Input{Shape $\Ss$ represented by $s$; precision parameter $\delta>0$; error probability $p$ }

\Output{Collection of $\delta$-symmetries $S = \Sym_\delta \Ss$ }

\medskip

Initialize $S = \emptyset$

\smallskip

Construct a $\rho=\frac{\delta}{2\TV}$-net $\Nn_\rho$ on $\OO(3)$

\smallskip

\While{not all symmetries have been detected}
{

\smallskip

Run Algorithm~\ref{algo:single-symm} on $\Nn_\rho$ to detect $\bb{R}$ with approximate distortion $d$
\smallskip

\lIf{$d > 1.5\cdot\delta$}{stop}

\eIf{$\bb{R}$ is a rotation (and not a reflection)}{
   Let $\mathcal{X}$ be the set of $n$-fold symmetries along its axis, with distortion $\le\delta\quad$ ($n=2,\cdots,N$)
   }{
   Let $\mathcal{X} = \{\bb{R}\}$
  }

{Add $\mathcal{X}$ to $S$ and remove from $\Nn_\rho$ a fixed-sized neighborhood of each symmetry in $\mathcal{X}$}

\smallskip
}

\smallskip

Return all detected transformations $\bb{R}$ \\
\hrulefill

\medskip

\caption{Detection of all approximate symmetries.  \label{algo:all-symm} }
\end{algorithm}

\subsection{Manipulating the shape complexity} \label{sec.shape.complex}

For a fixed precision $\delta$, the complexity of our symmetry detection algorithm is governed by the term $(r\cdot\TV)^3$, which is the cube of the shape \emph{complexity factor} $\CC = r\cdot\TV$, a unit-less quantity that resembles the isoperimetric quotient and describes the geometric complexity of the function $s$ representing the shape.

Through the total variation of $s$, $\CC$ depends on the function representing $\Ss$ and not directly on $\Ss$ itself. This leads to the important issue of designing representation functions for shapes that minimize the computation complexity.

To this end, we suggest controlling the shape complexity using \emph{truncated signed distance function} (TSDF) representations.
The advantages of doing so are two-fold: First, the TSDFs produce smoother shape representations, which lead to faster running time.
Additionally, the resulting representation values have lower variance (as a result of increased smoothness), allowing use of tighter bounds than the one stated in Proposition~\ref{prop:proptest}, which only assumes that the summands are bounded but does not consider their variance.

Nevertheless, these advantages come at the cost of the enhancement of  thin structures (possibly amplifying the effect of shot-noise) as well as the possible loss of discriminativity, due to smoothing of fine details. The effects (noise amplification and loss of discriminativity) are studied below. Also, as part of our large-scale experiment in Section~\ref{sec.largescle}, we analyze the influence of these choices on discriminativity.

For a given (truncation) constant $K$, we define the Euclidean TSDF by
\begin{eqnarray}
d_K(\bb{x},\partial \Ss) &=& \min\{K, \max\{-K, d(\bb{x},\partial \Ss) \}\}
\end{eqnarray}
where $d(\bb{x},\partial \Ss)$ is the signed distance map from the boundary $\partial \Ss$ of the shape $S$.
Since $d(\bb{x},\partial \Ss)$ satisfies the eikonal equation $\| \nabla d(\bb{x},\partial \Ss) \| = 1$ almost everywhere,
the total variation of $d_K(\bb{x},\partial \Ss)$ is bounded by
\begin{eqnarray}
\TV(d_K) &\le & \frac{c\cdot K \cdot \Area\, \partial \Ss}{\Vol\,B_{r+K}}, 
\end{eqnarray}
where $c$ is a constant.

We construct a family of functions
\begin{eqnarray}
s_K(\bb{x}) &=& \frac{1}{2K} d_K(\bb{x},\partial\Ss) + \frac{1}{2}
\end{eqnarray}
with the image in $[0,1]$, whose $\frac{1}{2}$-sub-level set is $\Ss$.
For $K=0$, $s_0$ is simply the binary indicator function.
Denoting $\TV_K = \TV(s_K)$ and $\CC_K = \CC(s_K)$, we observe that
\begin{eqnarray}
\CC_K & \le & (r+K) \cdot \TV_K \leq (r+K)  \frac{ \TV_0 \cdot \Vol\,B_r}{\Vol\,B_{r+K}} \nonumber\\
&=& \CC_0 \left(\frac{r}{r+K}\right)^2,\label{eq:shape-complex}
\end{eqnarray}
which for $K \gg r$ becomes $\CC_K \sim \mathcal{O}(K^{-2})$. Therefore, from the point of view of the complexity factor alone,
it is advantageous to increase $K$ without limits.

However, a large $K$ has a negative impact on the noise resilience of the symmetry detection algorithm.
To visualize this, assume that the shape is almost perfectly symmetric,
such that under a transformation $\bb{R} \in \OO(3)$, $s_0$ and $s_0 \bb{R}$ match except for on a small ball $B_\epsilon$ resulting from noise.
Therefore, using $s_0$, $\bb{R}$ has the distortion of $\delta_0 = \dis_s\,\bb{R} = \Vol\,B_\epsilon / \Vol\,B_r$. Increasing $K$ yields
\begin{eqnarray}
\delta_K &=& \frac{\Vol\,B_{\epsilon+K}}{\Vol\,B_{r+K}} = \left(\frac{\epsilon+K}{r+K}\right)^3,
\end{eqnarray}
which for $K \gg r$ becomes $\delta_K \sim 1$, amplifying the noise to unreasonable proportions.

\begin{figure*} [t]
\begin{center}
\addtolength{\tabcolsep}{-3pt}
    \begin{tabular}{c c c c c c}
      \hspace{-9pt}
      \includegraphics[width = 0.12\textheight]{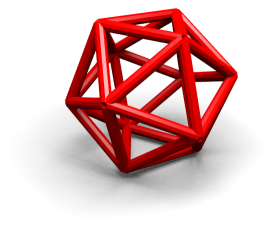} &
      \includegraphics[width = 0.09\textheight]{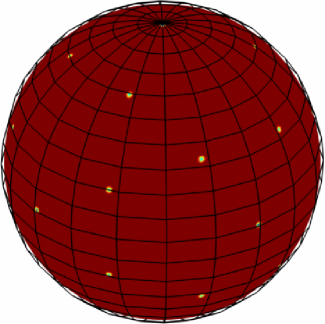} &
      \includegraphics[width = 0.09\textheight]{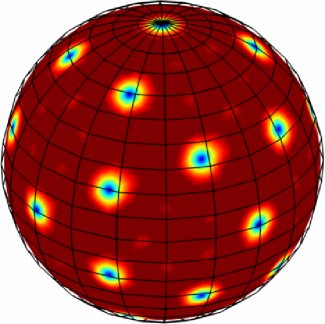} &
      \includegraphics[width = 0.09\textheight]{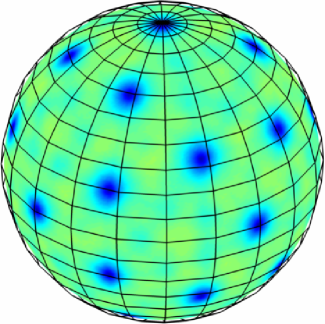} &
      \includegraphics[width = 0.09\textheight]{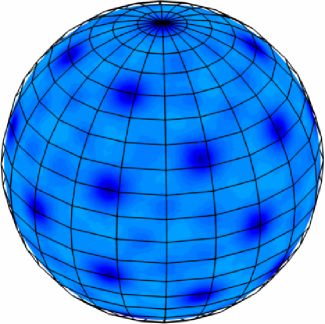} &
      \includegraphics[width = 0.18\textheight]{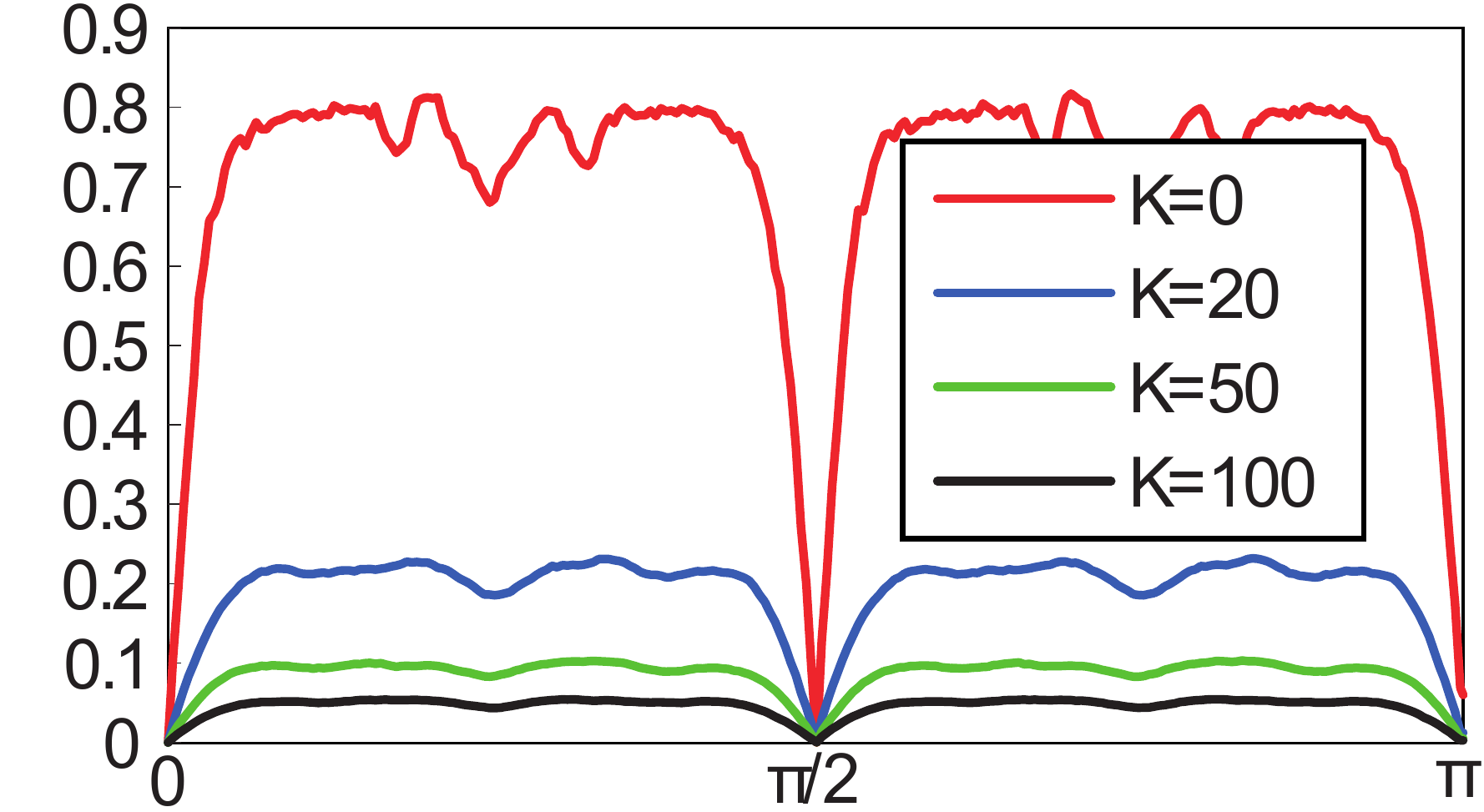}\\
      \hspace{-9pt}
      \includegraphics[width = 0.12\textheight]{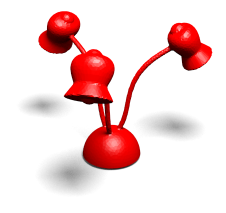} &
      \includegraphics[width = 0.09\textheight]{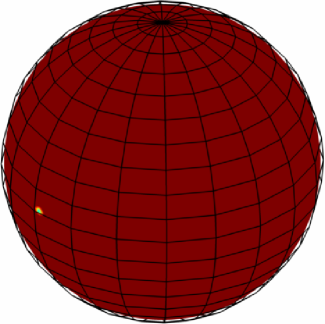} &
      \includegraphics[width = 0.09\textheight]{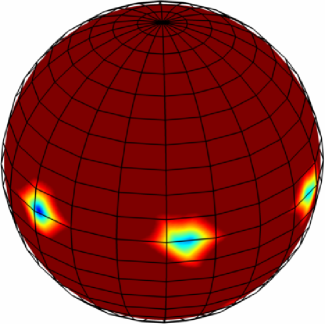} &
      \includegraphics[width = 0.09\textheight]{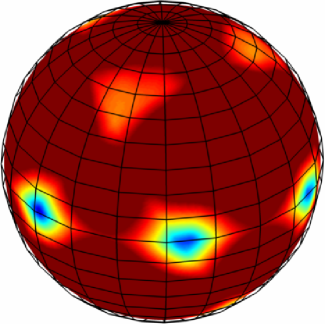} &
      \includegraphics[width = 0.09\textheight]{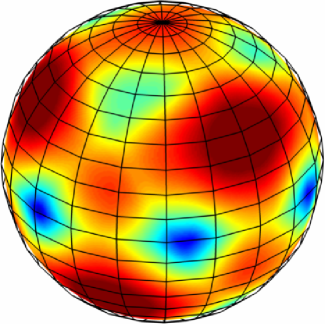} &
      \includegraphics[width = 0.18\textheight]{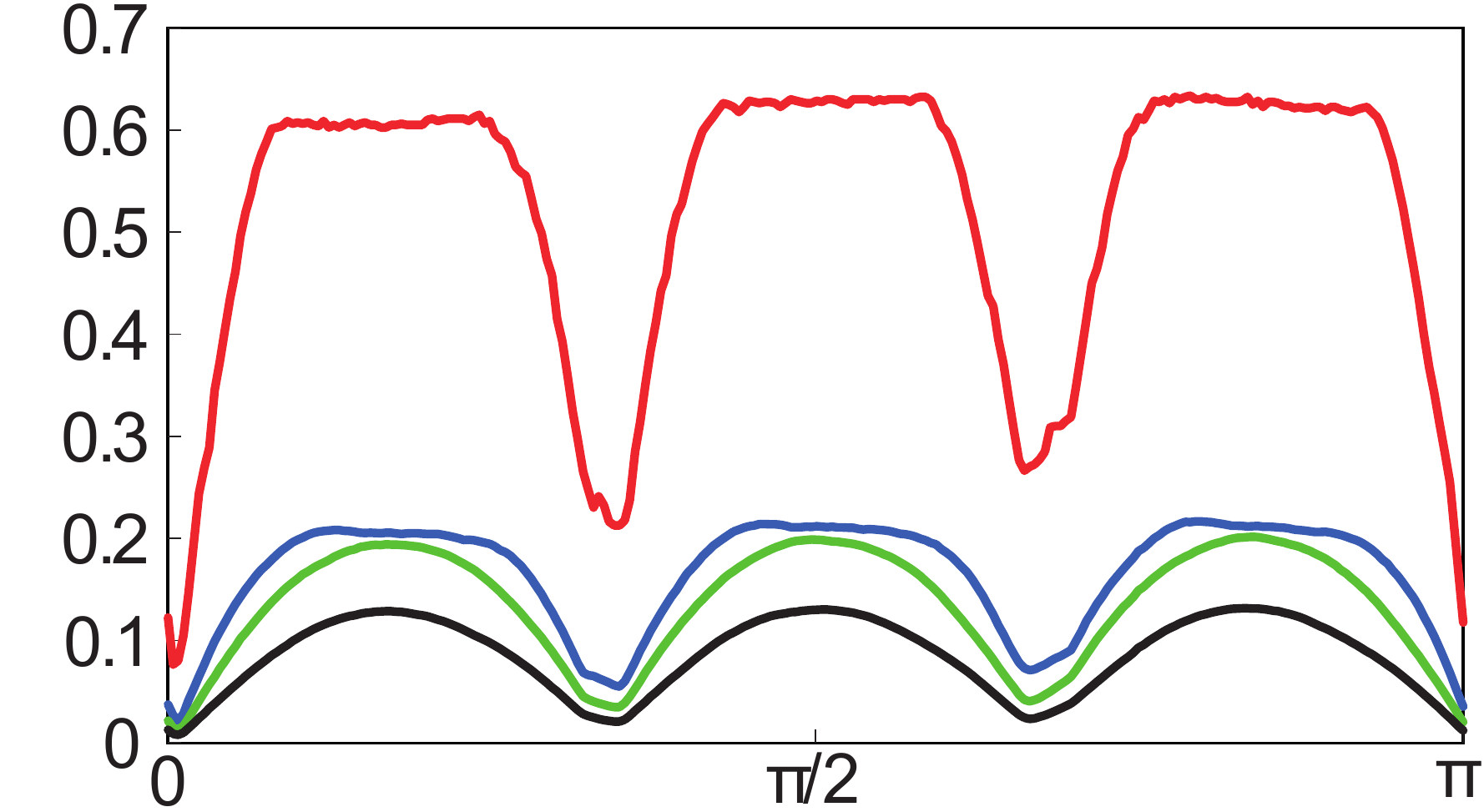}\\
      \hspace{-9pt}
      \includegraphics[width = 0.12\textheight]{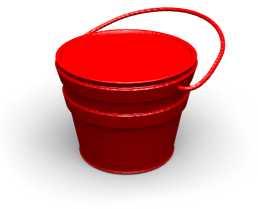} &
      \includegraphics[width = 0.09\textheight]{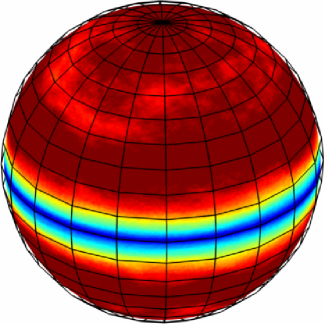} &
      \includegraphics[width = 0.09\textheight]{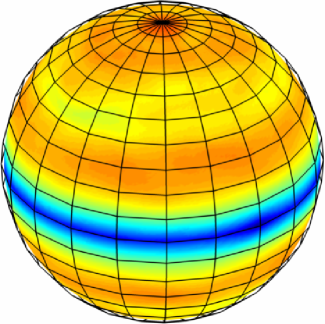} &
      \includegraphics[width = 0.09\textheight]{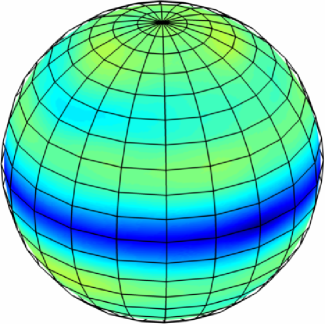} &
      \includegraphics[width = 0.09\textheight]{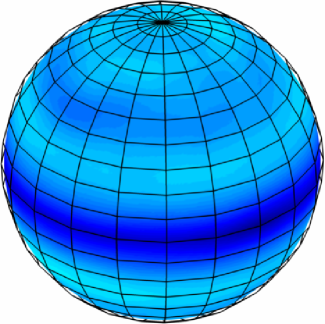} &
      \includegraphics[width = 0.18\textheight]{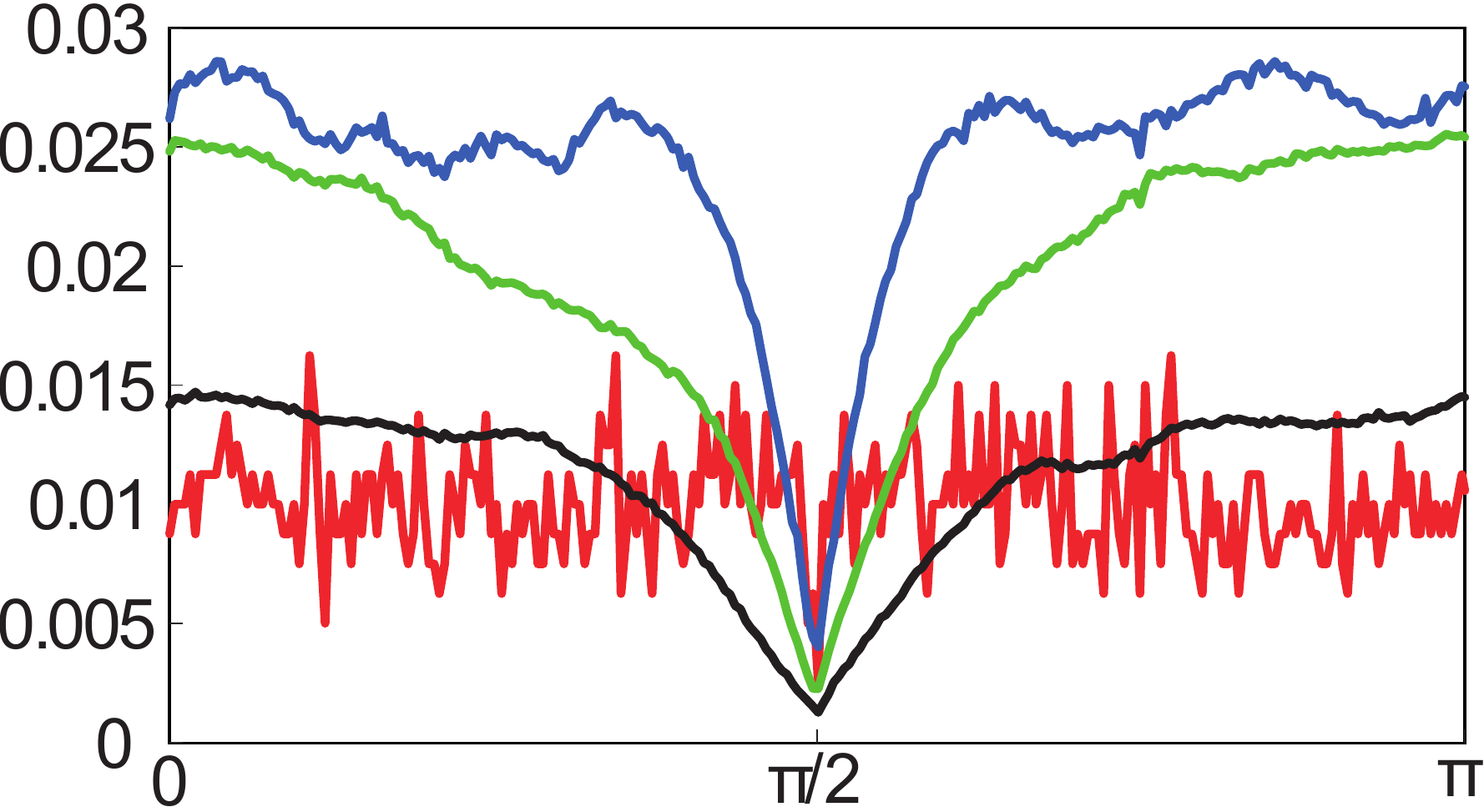}\\
      \hspace{-9pt}
(a) shape & (b) K=0 (binary) & (c) K=20\% & (d) K=50\% & (e) K=100\% & (f) 1D profile along equator \\
        \end{tabular}
        \vspace{-2pt}
        \caption{\textbf{Distortion maps for different truncation levels.} We consider only reflection symmetries in this illustration. Each row shows one example in the following format: (a) the shape (b)-(e): distortion levels of the function $s_K$ for different truncation values $K$. Color-coding ranges from 0 (blue) to a clipped value of 0.2 (red) at each location and represents the respective distortion of the planar reflection symmetry, whose normal passes through the point. (f) a section of the distortion along the equator of the sphere. See text for interpretation.}
        \label{fig.LK.Features}
    \end{center}\vspace{-16pt}
\end{figure*}


The same effect also decreased the sensitivity to fine features. Suppose two approximate symmetries $R_1$ and $R_2$ maintain the shape invariant except for on small balls $B_{\epsilon}$ and $B_{2\epsilon}$, respectively. The corresponding distortions when using $s_0$ are\vspace{-12pt}

\begin{equation*}
  \delta^1=\frac{\Vol\,B_{\epsilon}}{\Vol\,B_{r}} = \left(\frac{\epsilon}{r}\right)^3\; , \quad
  \delta^2=\frac{\Vol\,B_{2\epsilon}}{\Vol\,B_{r}} = \left(\frac{2\epsilon}{r}\right)^3=8\delta^1 \; .
\end{equation*}

That is, $R_2$ is an order of magnitude "worse" than $R_1$, which is typically an easily detectable situation.

When increasing $K$ ,the discriminativity, viewed as the ratio
\begin{equation}
  \frac{\delta^2_K}{\delta^1_K} = {{\left(\frac{2\epsilon+K}{r+K}\right)^3}}\mathlarger{\mathlarger{\mathlarger{/}}}{{\left(\frac{\epsilon+K}{r+K}\right)^3}} = \left(\frac{2\epsilon+K}{\epsilon+K}\right)^3
\end{equation}
approaches $1$ for $K>>\epsilon$, meaning that features of size $K$ are smoothed out.

\section{Experiments and applications}

The code used to generate the reported results can be downloaded from the project webpage \cite{ShapeSymmetryWebpage}.

\paragraph*{Data-sets} The main data-set we work with is the COSEG data-set~\cite{wang2012active}, which includes (among other) 190 shapes belonging to eight categories, that were originally purported for the evaluation of segmentation algorithms. While the shapes were created using CAD tools, they are not completely synthetic in their nature. Specifically, while most of the shapes have at least one kind of symmetry, in the vast majority of the cases the symmetry is far from
being perfect, which makes its detection challenging. We first rasterize a randomly rotated
version of each shape into a cartesian voxelized volume, where the maximal dimension is taken to be 160 voxels. Since all COSEG shapes are vertically aligned, a random rotation disables the advantage of any specific sampling location. Then, we center the shape around its centroid and measure its support. Finally, we pad and crop the volume to a cube, with the side length twice the shape support radius $r$. This guarantees that the shape remains within the volume under arbitrary rotations and reflections. The final volume dimensions are around $200^3$.

We also created a small data-set of shapes that have complex symmetry groups. These include the icosahedron and the dodecahedron (See Figure~\ref{fig.runtimes2}\; for an illustration). Finding \emph{all} the symmetries of such shapes is computationally challenging.
The third type of data we use is a volumetric scalar MRI image taken from~\cite{Cocosco97brainweb:online}.

\paragraph*{Symmetries and their notations} We seek to find approximate symmetries of different kinds. The first kind are \emph{planar reflections}, around a generally oriented plane which passes through the centroid. We denote such a symmetry by \textbf{REFL} and visualize it as a transparent plane. The second kind are $t$-fold rotations around an axis that passes through the centroid (where we search for $t$ between $2$ and $20$). Such a symmetry (or a set of symmetries) includes rotation symmetries of the set of angles $\{2\pi i/t\}$ for $i = 1,...,t-1$. We denote such a symmetry by \textbf{t-fold-ROT} and visualize its rotation axis in red. The third kind, axial-symmetries, are fully-continuous rotation symmetries around some axis. We denote these by \textbf{CONT} and visualize them using a magenta colored axis.

\paragraph*{Algorithm settings and implementation details} We represent a shape $s_K$ by applying a TSDF, where the truncation parameter $K$ is chosen adaptively such that the total variation $\TV$ of the shape is approximately $3/r$, as detailed in Section~\ref{sec:auto_K}. When running the main algorithm (Algorithm~\ref{algo:all-symm}), we aim for high precision, which translates into invoking the (single symmetry detection) Algorithm~\ref{algo:single-symm} with low values of the precision parameter $\delta$. For efficiency, we run Algorithm~\ref{algo:single-symm} in a branch-and-bound manner, which begins with an initial (coarse) net defined by $\delta=0.25$ and iteratively increases resolution only in "promising" regions of the transformation space, finally reaching the desired resolution. Note that this can be done as in~\cite{FastMatchPaper}, based on our net construction, while keeping the theoretical guarantees. In addition, after the detection of each symmetry, we carve out all transformations whose symmetry axis is less than $10^{\circ}$ from that of the detected axis and then repeat the search for additional symmetries. All experiments were run on a 2.70GHz machine, with 8GB RAM.
Our timings throughout (excluding Table~\ref{tbl-Kazhdan-all} and \ref{tbl-Kazhdan}) do not include a $0.8$ seconds pre-processing time for the TSDF computation, per shape.

\subsection{The influence of truncation} \label{sec:truncation-influence}

\begin{figure} [t!]
  \begin{center}
    \includegraphics[width=1\linewidth]{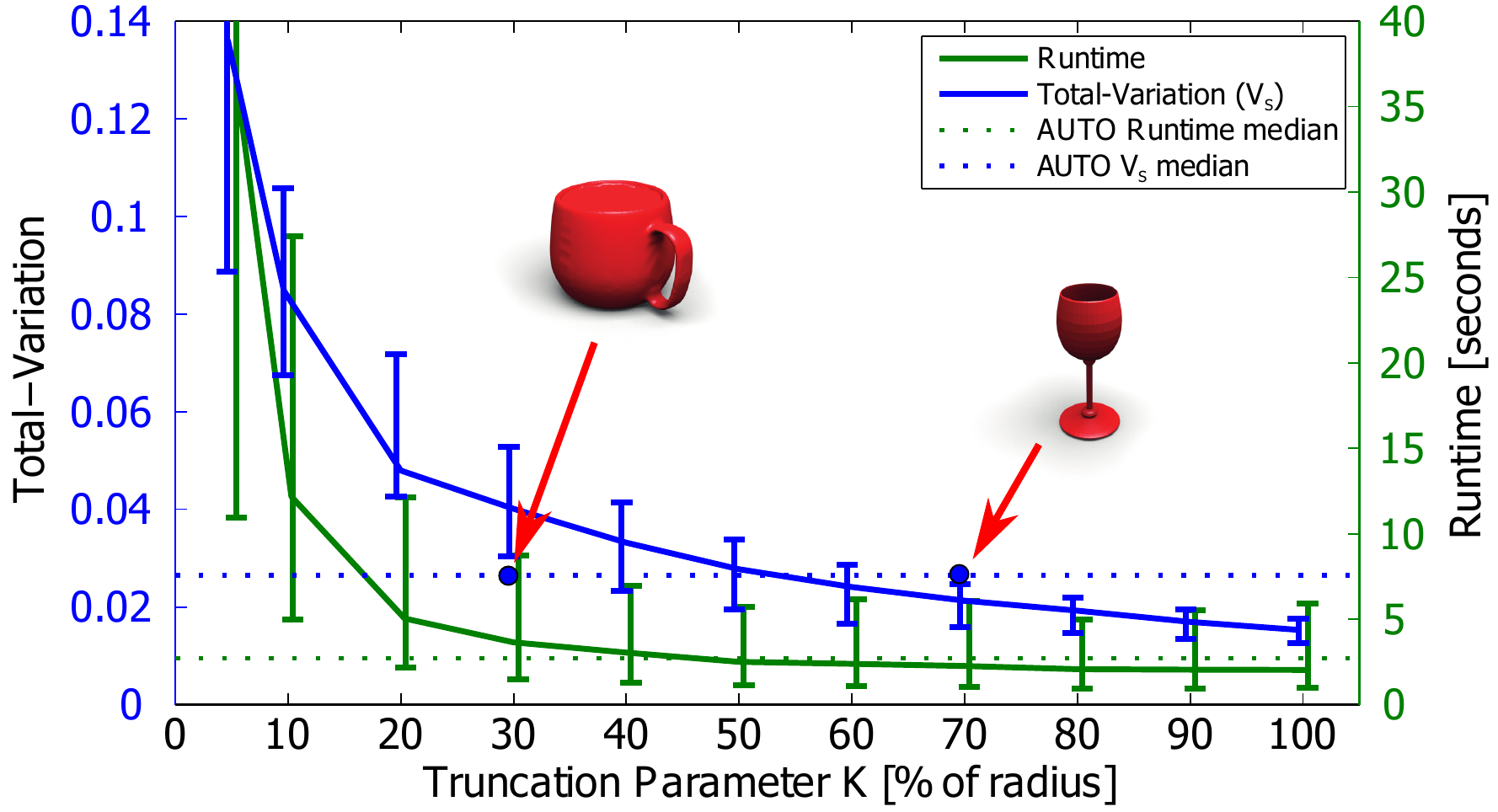}
\caption{\textbf{The empirical dependence of $\TV$ and run time on $K$.}
Large $K$ leads to lower total-variation $\TV$.
Plotted are the median of the total-variation (and of the runtime in seconds),
over the entire COSEG data-set.
Error bars are $10^{th}$ and $90^{th}$ percentiles.
See Section~\ref{sec:truncation-influence} for details.
}\vspace{-4pt}
\label{fig.KvsAV}
  \end{center}
\end{figure}

Figure~\ref{fig.LK.Features} shows intermediate results of running our algorithm when a shape is represented with the binary indicator function, $s_0$, as well as with $s_K$ with various truncation values $K$, on a variety of shapes.
Each row depicts a different shape, followed by color-maps of the different distortion levels on a hemisphere (The hemisphere has the same orientation of the shape). For simplicity we focus only on planar reflective symmetries. The color-coding at each location on the hemisphere represents the distortion of the planar reflection symmetry, whose normal passes through the point. Each column of (b)-(e) shows such a sphere, for $K$ taken to be $0\%$, $20\%$, $50\%$, and $100\%$ of the shape radius. The last column (f) shows a 1D profile of the map around the equator of the sphere.
As can be seen, when $K=0$ (the binary indicator case) the distortion is relatively volatile. Increasing $K$ decreases the total variation (and shape complexity factor, as in Eq. \eqref{eq:shape-complex}) and therefore makes the distortion map much smoother (see Proposition~\ref{prop:bound}). As a result, the sampling rate required by the algorithm can be decreased.
Notice that in the case of the lamp (second row) the distortion of the binary indicator function (i.e. $K=0$) is high even around approximate symmetries, because the shape is not perfectly symmetric. An increase in $K$ is necessary is such cases (see discussion in Section~\ref{sec.largescle}).
Finally, the binary indicator function is not affected much by the handle of the bucket (last row) because this is a very thin structure. The detection of the exact symmetry would require extremely fine sampling, while increasing $K$ increases the sensitivity to fine features.

\begin{figure}[t]
\begin{center}
\addtolength{\tabcolsep}{-5pt}
    \begin{tabular}{c c c c c}
      \rotatebox{90}{\textbf{\small{\emph{'Chairs 102'}}}} \hspace{3pt} &
      \includegraphics[width = 0.105\textwidth]{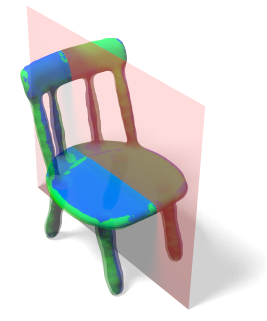} &
      \includegraphics[width = 0.105\textwidth]{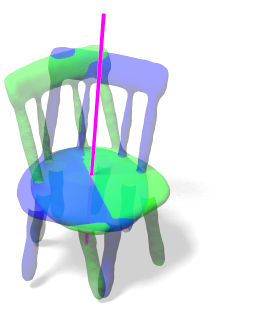} &
      \includegraphics[width = 0.105\textwidth]{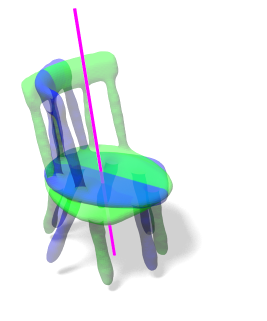} &
      \includegraphics[width = 0.105\textwidth]{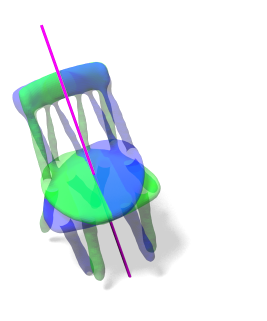}\vspace{-4pt}\\
     & \textbf{\tiny{REFL (0.014)}}&\textbf{\tiny{CONT (0.21)}}&\textbf{\tiny{CONT (0.2)}}&\textbf{\tiny{CONT (0.2)}}\\
      \rotatebox{90}{\textbf{\small{\emph{'Irons 103'}}}} \hspace{3pt} &
      \includegraphics[width = 0.105\textwidth]{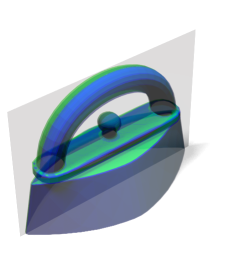} &
      \includegraphics[width = 0.105\textwidth]{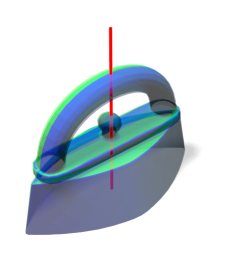} &
      \includegraphics[width = 0.105\textwidth]{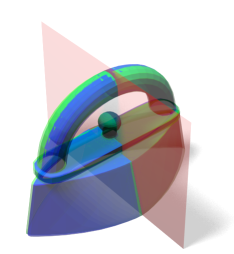} &
      \includegraphics[width = 0.105\textwidth]{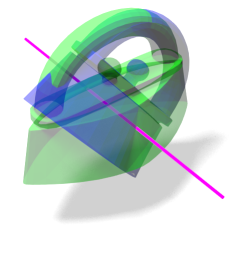}\vspace{-4pt}\\
     & \textbf{\tiny{REFL (0.012)}}&\textbf{\tiny{2-ROT (0.015)}}&\textbf{\tiny{REFL (0.018)}}&\textbf{\tiny{CONT (0.17)}}\\
      \rotatebox{90}{\textbf{\small{\emph{'Vases 826'}}}} \hspace{3pt} &
      \includegraphics[width = 0.105\textwidth]{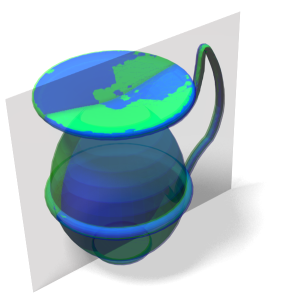}&
      \includegraphics[width = 0.105\textwidth]{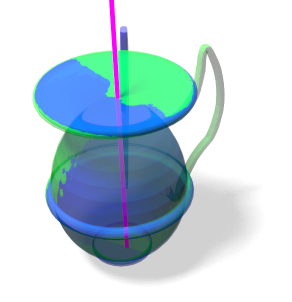} &
      \includegraphics[width = 0.105\textwidth]{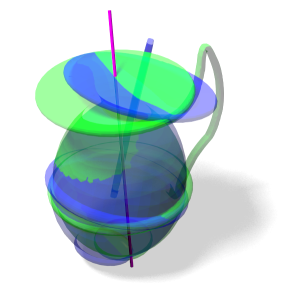} &
      \includegraphics[width = 0.105\textwidth]{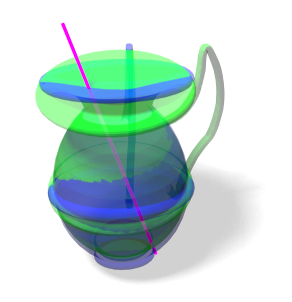}\vspace{-4pt}\\
    & \textbf{\tiny{REFL (0.001)}}&\textbf{\tiny{CONT (0.03)}}&\textbf{\tiny{CONT (0.086)}}&\textbf{\tiny{CONT (0.087)}}\\
       \end{tabular}
        \vspace{2pt}
        \caption{\textbf{Four best symmetries of COSEG shapes.} Three typical examples from the data set with different numbers of symmetries of different quality. The original shape is shown in green and the shape mapped to by the detected transformation is shown in blue. The chair shape has one clear reflection; three other transformations have higher distortion (reported in parentheses). The iron has three good symmetries (two reflections and a 2-fold rotation), while the vase has one nearly perfect reflection symmetry, and another rotation symmetry that is inexact because of the handle.}
        \label{fig.coseg.examples}
    \end{center}\vspace{-18pt}
\end{figure}

\subsection{Automatic selection of the truncation parameter $K$} \label{sec:auto_K}

It is desirable to make the runtime of Algorithm~\ref{algo:single-symm} depend only on the precision parameter $\delta$, and not on the properties of the shape itself.
Recall that the number of samples (net size) depends on the total variation $\TV$, which in turn depends on the shape representation.
Therefore, we would like to automatically choose a value of $K$ for each shape, that will produce $\TV$ proportional to $1/r$. This will make the shape complexity factor $\CC$ relatively constant (since $\CC = r\cdot \TV$) and therefore ensures constant runtime for a wide range of values of $\delta$.
To this end, we measured the empirical dependence of the total variation on the truncation level $K$, over the entire COSEG data-set.
Figure~\ref{fig.KvsAV} shows the median, 10th and 90th percentiles of $\TV$ (in blue) and runtimes (in green, for a single detection per shape). Observe that increasing $K$ reduces both runtime and the total variation of the shape, in accordance with the bound stated in Section~\ref{sec.shape.complex}.

The automatic choice of $K$ for an unseen shape, is determined by initially calculating its $\TV$ for an arbitrary value of $K$ and then by improving the choice using a binary search, according to the above empirical distribution. We empirically chose the goal of $\CC=3$, which implies $\TV=3/r$.
Figure~\ref{fig.KvsAV} also shows two specific shapes. For the shape `\emph{goblets 32}' (right) whose shape complexity $\CC$ is large, we use a rather high truncation of $K=0.7r$ in order to reduce its $\TV$ drastically to the order of $3/r$. On the other hand, the shape `\emph{vases 801}' (left), which is more solid and has a lower shape complexity requires a much lower truncation of $K=0.3r$. The described automatic calculation of $K$, which we term 'AUTO', was used throughout the following experiments.

We empirically evaluated the influence of the truncation level $K$ on the discriminativity of the detection algorithm, as part of the large-scale experiment in the following section.

\subsection{Large-scale evaluation}\label{sec.largescle}


\begin{figure}[t]
  \begin{center}
    \includegraphics[width=0.95\linewidth]{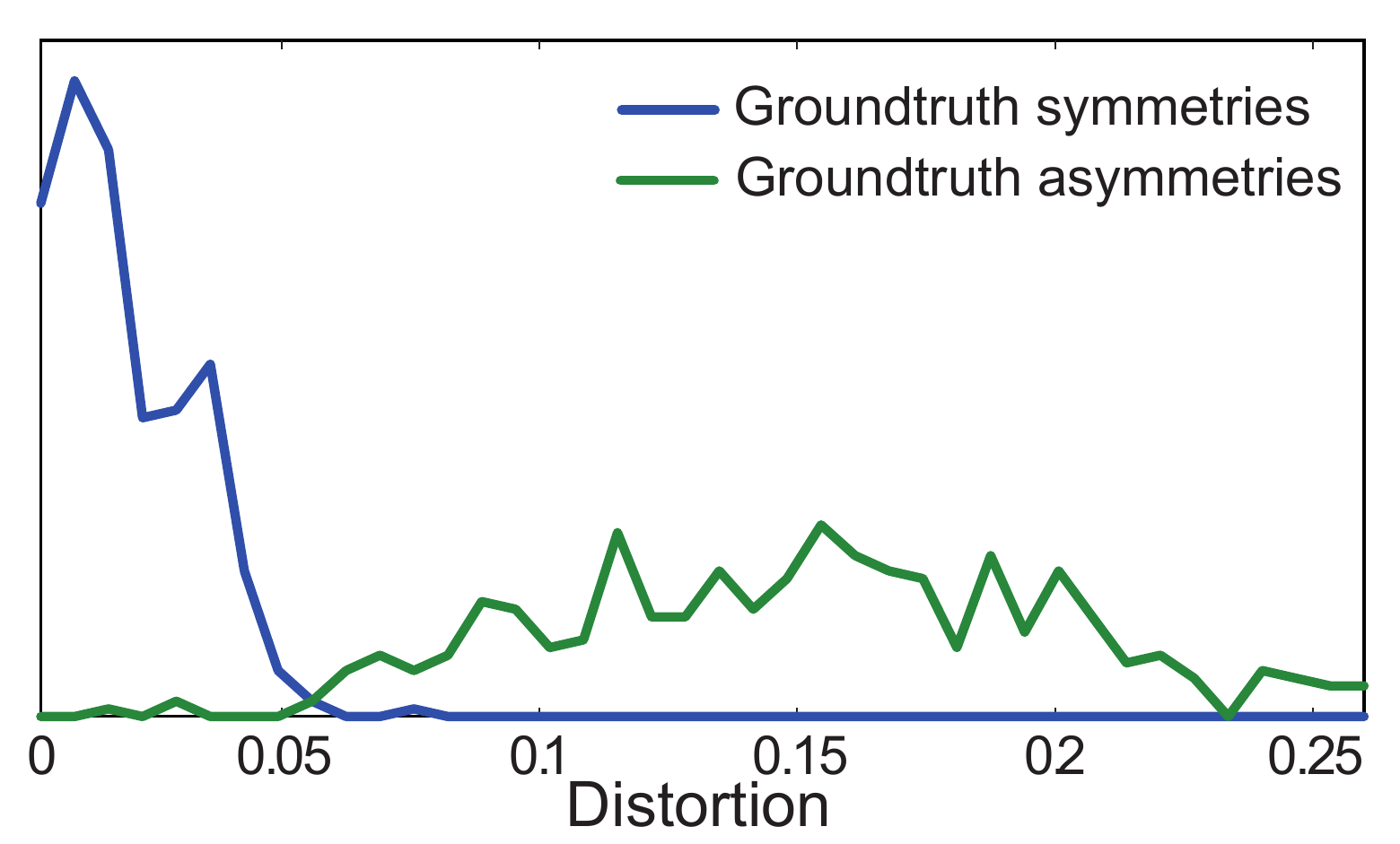} 
\vspace{-4pt}
\caption{\textbf{Distortion distribution of manually labeled results.} We manually labeled as "approximate symmetry" or "non-symmetry" the best four detections of the 190 COSEG shapes, and plotted the distortion distributions over these two sample sets. It turns out that these sets can be separated nicely by thresholding the distortion. This empirically established threshold is used as the stopping criterion in subsequent experiments.}
\label{fig.distributions}
  \end{center}\vspace{-18pt}
\end{figure}

After establishing the accuracy and runtime complexity of our algorithm, we tested it on the complete COSEG dataset. To the best of our knowledge, we are the first to report symmetry detection results on a data-set of this scale. At the first stage, we ran the algorithm to detect the four best symmetries per shape giving a total of $760=190\cdot4$ symmetries. We did so (even though each shape might have less or more than 4 approximate symmetries) in order to investigate the option of automatically detecting when the returned symmetry is indeed an approximate symmetry. See Figure~\ref{fig.coseg.examples} for the four first detected symmetries (along with their distortion levels) for three such example shapes. As an example, the 'Chairs 102' shape has only one approximate symmetry, a planar reflection, which we detect first and which has low distortion of 0.014. The three remaining detected symmetries, which can not be considered approximate symmetries, have distortions levels of over 0.2.

We manually labeled all the `potential' symmetries returned by our algorithm as either "approximate symmetry" or "non-symmetry" and plotted a histogram of the respective distortions. As can be seen in Figure~\ref{fig.distributions}, our distortion measure is fairly invariant to the shape and a global threshold of $\dis=0.05$  can be used to determine if the detected symmetry is indeed a meaningful symmetry. We then ran the algorithm again, with the stopping criterion defined by this threshold. The algorithm found a total of 463 symmetries in the 190 shapes. Please refer to the supplementary material (availiable also on \cite{ShapeSymmetryWebpage}) for the complete set of detections. In Figure~\ref{fig.examples.good}, we show the approximate symmetries that were detected for eight representative shapes. Notice, that while some of the shapes have almost perfect symmetries (camel, lamp, candelabras), some others have (possibly, in addition to a perfect symmetry) some symmetries that are only approximate (e.g.  cup and guitar). In the latter cases, our visualization shows how well each part of the shape undergoes the symmetry (see cup handle and guitar neck).
%

\begin{figure}[t]
\begin{center}
\addtolength{\tabcolsep}{-5pt}
    \begin{tabular}{c c c c}
\includegraphics[height =  0.11\textwidth]{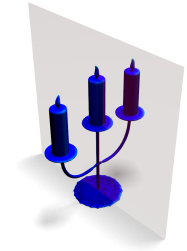}&
\includegraphics[height =  0.11\textwidth]{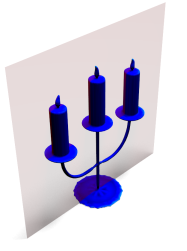}&
\includegraphics[height =  0.11\textwidth]{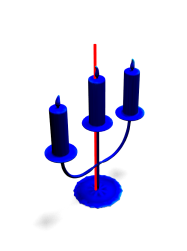}&
\includegraphics[height =  0.11\textwidth]{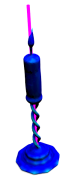}\vspace{-2pt}\\
\textbf{\footnotesize{\emph{ Candelab. 27}}}&\textbf{\footnotesize{\emph{ Candelab. 27}}}&\textbf{\footnotesize{\emph{ Candelab. 27}}}&\textbf{\footnotesize{\emph{ Candelab. 28}'}}\\
\textbf{\tiny{REFL (0.004)}}&\textbf{\tiny{REFL (0.010)}}&\textbf{\tiny{2-ROT (0.012)}}&\textbf{\tiny{CONT (0.011)}}\vspace{1pt}\\
\includegraphics[height =  0.11\textwidth]{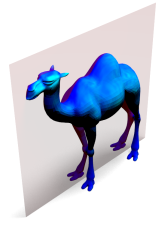}&
\includegraphics[height =  0.11\textwidth]{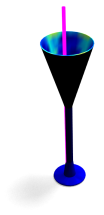}&
\includegraphics[height =  0.11\textwidth]{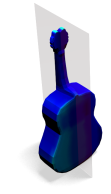}&
\includegraphics[height =  0.11\textwidth]{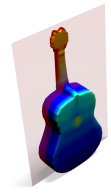}\vspace{-2pt}\\
\textbf{\footnotesize{\emph{ Fourleg 393}}}&\textbf{\footnotesize{\emph{ Goblets 12}}}&\textbf{\footnotesize{\emph{ Guitars 427}}}&\textbf{\footnotesize{\emph{ Guitars 427}}}\\
\textbf{\tiny{REFL (0.008)}}&\textbf{\tiny{CONT (0.010)}}&\textbf{\tiny{REFL (0.011)}}&\textbf{\tiny{REFL (0.024)}}\vspace{1pt}\\
\includegraphics[height =  0.11\textwidth]{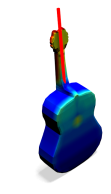}&
\includegraphics[height =  0.11\textwidth]{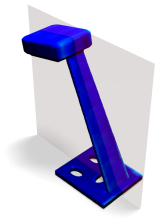}&
\includegraphics[height =  0.11\textwidth]{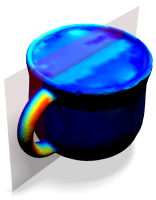}&
\includegraphics[height =  0.11\textwidth]{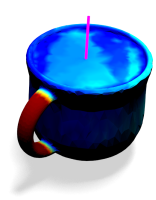}\vspace{-2pt}\\
\textbf{\footnotesize{\emph{ Guitars 427}}}&\textbf{\footnotesize{\emph{ Lamps 18}'}}&\textbf{\footnotesize{\emph{ Vases 817}'}}&\textbf{\footnotesize{\emph{ Vases 817}'}}\\
\textbf{\tiny{2-ROT (0.028)}}&\textbf{\tiny{REFL (0.008)}}&\textbf{\tiny{REFL (0.007)}}&\textbf{\tiny{CONT (0.023)}}\vspace{1pt}\\
       \end{tabular}
        \vspace{-2pt}
        \caption{\label{fig.examples.good}
        \textbf{Detected approximate symmetries on representative COSEG shapes.} Several examples - 8 shapes and the 12 symmetries we detected for them, using the threshold from Figure~\ref{fig.distributions} as the largest admissible distortion. See the supplementary material for all 463 symmetries detected on the entire dataset.}
    \end{center}
    \vspace{-18pt}
\end{figure}

\begin{figure}[b]
\vspace{-6pt}
  \begin{center}
    \includegraphics[width=0.9\linewidth]{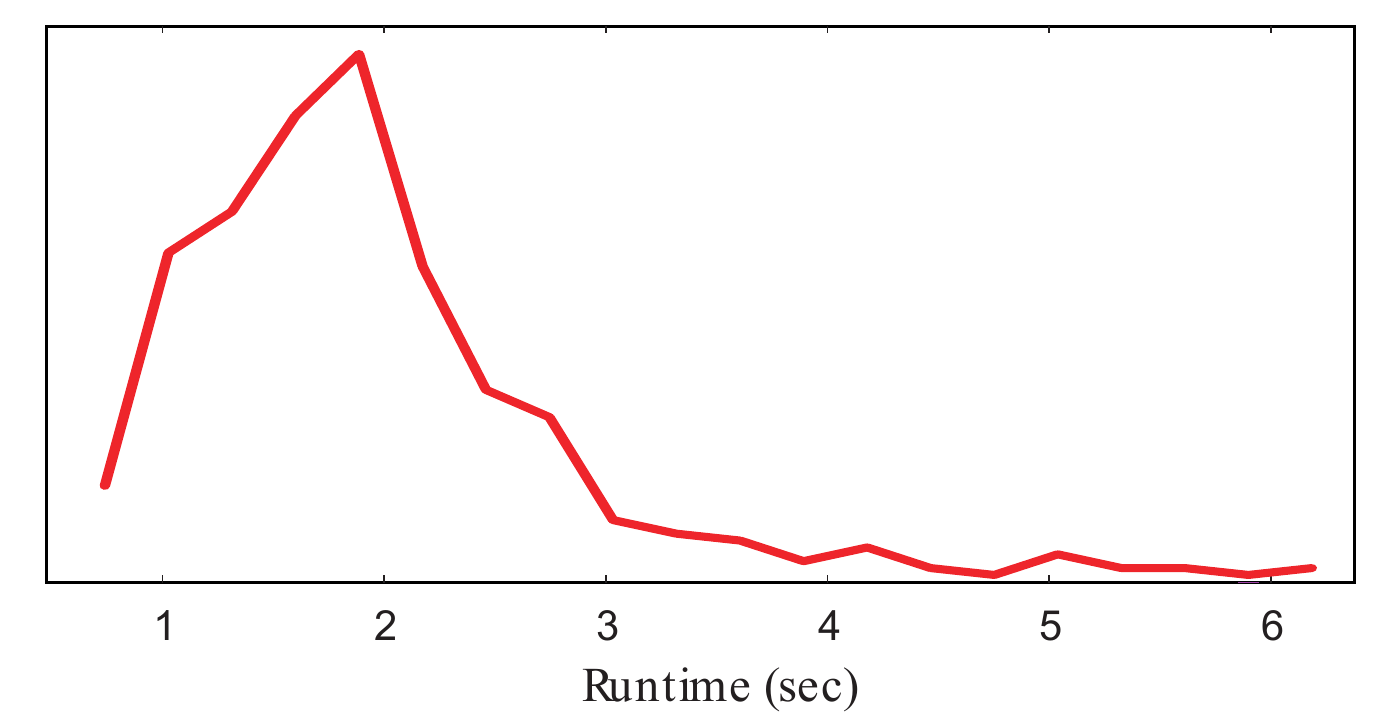} \vspace{-8pt}
\caption{\textbf{Distribution of runtime per detected symmetry}, on the 463 detected symmetries in the COSEG data-set.}
\label{fig.runtimes1}
  \end{center}\vspace{-15pt}
\end{figure}

\noindent\textbf{Discriminativity analysis}\quad
The use of TSDF representations significantly improves the runtime complexity of the algorithm, as was shown in Figure~\ref{fig.KvsAV}. We now turn to test the qualitative consequences of manipulating the truncation parameter $K$ (as well as of the automatic choice of $K$).
While the large-scale experiment was done using the automatic selection of $K$ (described in Section~\ref{sec:auto_K}), here we experiment also using a fixed set of truncation levels $K$. The results are summarized in Figure~\ref{fig.nDetvsK}, where we count for each representation the number of detected approximate-symmetries as well as the number of false detections. As expected, the increase in $K$ comes at a certain loss of discriminativity, as the ratio between true and false detections slightly deteriorates. The increase in the number of detected approximate symmetries, up to a certain level of $K$, is due to the rejection of approximate symmetries by the less smooth representations (i.e. lower $K$'s). Most noticeably, our automatic selection mode outperforms any fixed selection of $K$ in terms of the number of true detections and the ratio of false detections as well as in terms of runtime (see Figure~\ref{fig.KvsAV}).

\subsection{Comparison with Kazhdan {\em et al.}\cite{kazhdan2004symmetry}}
\label{sec:kazhdan-comparison}

As mentioned in Section~\ref{sec:prior-work}, the work of Kazhdan {\em et al} \cite{kazhdan2004symmetry} bares some resemblance to the present work, mainly because it also directly evaluates many transformations using the measure of \cite{zabrodsky1995symmetry} and therefore some of our bounds may apply to it. In spite of the similarities, there are some major differences which make the comparison difficult. The methods mainly differ in the choice of the set of transformations, as well as in the way by which they are evaluated.

The FFT-like approach in \cite{kazhdan2004symmetry} requires a regular $n\times n$ grid sampling of the latitude-longitude space of transformations, where $n$ is preferably a complete power of $2$. For this setting, the complexity of \cite{kazhdan2004symmetry} is $O(n^4)$ , which is dominated by an auto-correlation computation of order $O(n^4)$, followed by a computation of order $O(kn^2)$ for detecting $k$-fold symmetries for each $k$.
In comparison, our method's complexity, $O(n^3\cdot\epsilon^{-2}\log\frac{1}{p})$, scales more gracefully with $n$.

We first evaluated both methods on the task of detecting the best symmetry in $O(3)$ (including reflections and rotations of up to 8 folds) for the shapes in the COSEG data-set and compared several grid-size configurations of \cite{kazhdan2004symmetry} against our method.
The publicly available implementation of \cite{kazhdan2004symmetry} does not allow access to the $O(n^4)$ auto-correlation result and requires to recalculate it for each $k$-fold detection and therefore, for fairness of comparison, we report the average time over all $k$-fold computations for \cite{kazhdan2004symmetry}. The results summarized in Table \ref{tbl-Kazhdan-all} show that our method reaches lower distortion values even when compared to a grid of $256^2$, which takes much longer to evaluate.
Note that all distortion results are calculated on the original shape representation, following \cite{zabrodsky1995symmetry}.


\begin{table}[h!]
  \centering
\addtolength{\tabcolsep}{-2pt}
\begin{tabular}{ crccc }
\parbox[c]{55pt}{algorithm} &
\parbox[r]{18pt}{grid size} &
\parbox[c]{30pt}{number of trans.} &
\parbox[c]{30pt}{distortion \cite{zabrodsky1995symmetry}} &
\parbox[c]{20pt}{runtime [sec]} \vspace{4pt}\\
 \hline
 \hline
\multirow{4}{*}{\parbox[c]{60pt}{Kazhdan \\ \emph{et al.} \\ \cite{kazhdan2004symmetry}}}
    & $32^2$  & $1,024$   & $0.160$  & $\mathbf{0.17}$ \\
    & $64^2$  & $4,096$   & $0.087$  & $0.42$ \\
    & $128^2$ & $16,384$  & $0.056$  & $3.42$ \\
    & $256^2$ & $65,536$  & $0.044$  & $30.56$ \\
  \hline
{\parbox[c]{70pt}{Proposed method} }
    &  $-\;\;\;$ & $106,054$  & $\mathbf{0.040}$  & $2.62$ \\
  \hline
\end{tabular}\vspace{-1pt}
\caption{\label{tbl-Kazhdan-all}
\textbf{Best symmetry detection.}
The table summarizes performance of two algorithms for best symmetry detection on the COSEG data-set.
Presented are median values for: (i) number of evaluated transformations (ii) symmetry measure \cite{zabrodsky1995symmetry} and (iii) run-time.
}
\vspace{-8pt}
\end{table}

\begin{figure}[t]
  \begin{center}
    \includegraphics[width=1\linewidth]{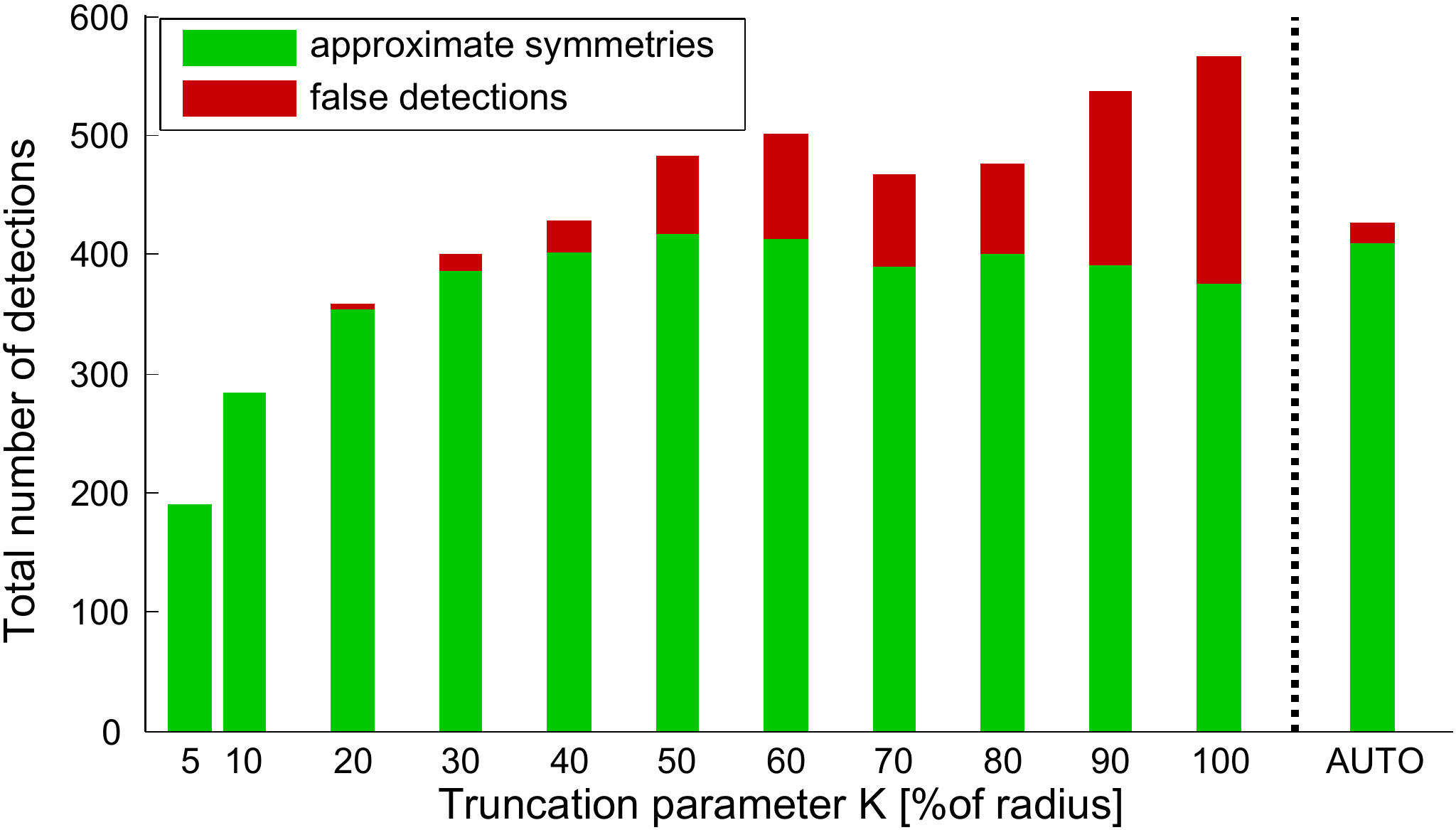}
\caption{\textbf{The empirical discriminativity vs. $K$.}
Low values of $K$ lead to
rejection of some approximate symmetries, while large ones reduce the discriminative power of our method. The automatic selection of $K$ leads to a favorable tradeoff.
See Section~\ref{sec.largescle} for details.
}\vspace{-14pt}
\label{fig.nDetvsK}
  \end{center}
\end{figure}

The sets of transformations used by both methods are quite similar when limited to reflection symmetries, therefore we were able to perform a more delicate comparison in this setting. We made several modifications to our method so it would be directly comparable to \cite{kazhdan2004symmetry}.
First, we set our net sizes to be equivalent to the grid sizes in the code of \cite{kazhdan2004symmetry}, rather than following Proposition~\ref{prop:sample-size} by using $\TV$ and $\delta$.
For the same reason, we disabled the branch-and-bound procedure, which allows reaching fine resolutions even with an initial coarse net. These modifications have a negative impact on our algorithm.
Note however that this comparison does not show the full effectiveness of \cite{kazhdan2004symmetry}, as the $O(n^4)$ auto-correlation preprocess (which is mandatory) is used only for the limited case of reflections.

We ran our algorithm in two configurations: one with the original binary shape ($K=0$), and another with automatic selection of $K$ (AUTO), as can be seen in Table~\ref{tbl-Kazhdan}.
Note that our net covers the transformation space uniformly, while the sampling used in \cite{kazhdan2004symmetry} (a product of equidistant samples in the azimuth-elevation space) is denser around the 'poles' and less around the 'equator'. This explains the slightly better distortions we obtain, for each fixed grid size.

\begin{table}[t]
  \centering
\begin{tabular}{ crccc }
\parbox[c]{40pt}{algorithm} &
\parbox[r]{20pt}{grid size} &
\parbox[c]{30pt}{number of trans.} &
\parbox[c]{40pt}{distortion \cite{zabrodsky1995symmetry}} &
\parbox[c]{20pt}{runtime [sec]} \vspace{4pt}\\
 \hline
 \hline
\multirow{4}{*}{\parbox[c]{40pt}{Kazhdan \\ \emph{et al.} \\ \cite{kazhdan2004symmetry}}}
    & $32^2$  & $1,024$   & $0.162$  & $\mathbf{0.19}$ \\
    & $64^2$  & $4,096$   & $0.085$  & $0.32$ \\
    & $128^2$ & $16,384$  & $0.057$  & $2.64$ \\
    & $256^2$ & $65,536$  & $0.044$  & $27.65$ \\
  \hline
\multirow{4}{*}{\parbox[c]{40pt}{Proposed \\ method \\ ($K=0$)} }
    & $ \sim 32^2$  & $1,031$   & $\mathbf{0.076}$  & $\mathbf{0.19}$ \\
    & $ \sim 64^2$  & $4,094$   & $0.055$  & $\mathbf{0.22}$ \\
    & $ \sim 128^2$ & $16,370$  & $0.045$  & $\mathbf{0.29}$ \\
    & $ \sim 256^2$ & $65,301$  & $0.044$  & $\mathbf{0.50}$ \\
  \hline
\multirow{4}{*}{\parbox[c]{40pt}{Proposed \\ method \\  (AUTO $K$)} }
    & $ \sim 32^2$  & $1,031$   & $\textbf{0.076}$  & $0.96^*$ \\
    & $ \sim 64^2$  & $4,094$   & $\mathbf{0.052}$  & $1.02^*$ \\
    & $ \sim 128^2$ & $16,370$  & $\mathbf{0.041}$  & $1.12^*$ \\
    & $ \sim 256^2$ & $65,301$  & $\mathbf{0.035}$  & $1.40^*$ \\
  \hline
\end{tabular}
\caption{\label{tbl-Kazhdan}
\textbf{Best reflection symmetry detection.}
The table summarizes performance of three algorithms for best reflection symmetry detection on the COSEG data-set.
See Table \ref{tbl-Kazhdan-all} and text for more details.
\quad (*) Runtimes for 'AUTO $K$' include a pre-processing of $\sim\;0.8$ seconds for the TSDF calculation.}
\vspace{-7pt}
\end{table}

A benefit of using the TSDF representation is evident when comparing the distortions achieved in the $K=0$ and the AUTO runs. As stated in Section~\ref{sec.shape.complex}, applying the TSDF lowers the variance of the shape representation, allowing use of tighter bounds than the one mentioned in proposition~\ref{prop:proptest} (which only assumes the individual summands are bounded and does not take advantage of their variance). As a result, better estimation of the distortion is achieved. 

\subsection{Complex symmetry groups}

We tested our algorithm on shapes with known complex symmetries. Cumulative runtimes for finding the entire set of symmetries are reported in Figure~\ref{fig.runtimes2}. In both cases (Dodecahedron and Icosahedron), the algorithm correctly detected exactly all 46 symmetries. The average symmetry detection time of the algorithm {\em decreases} as the number of symmetries of a shape increases because after each symmetry is detected we carve out its neighborhood.

Martinet \emph{et al.}~\cite{martinet2006accurate} evaluate their algorithm on the icosahedron as well.
The runtime we report here (45 seconds) is more than an order of magnitude smaller compared to the 50 minutes for detecting all symmetries of the icosahedron, reported in \cite{martinet2006accurate}, when treating the icosahedron as a single complex shape. The authors of \cite{martinet2006accurate} also report a runtime of 1 minute and 57 seconds when applying their 'constructive' method, which assumes that the icosahedron is given as a segmented set of 30 tiles.


\begin{figure}[t]
  \begin{center}
    \includegraphics[width=1\linewidth]{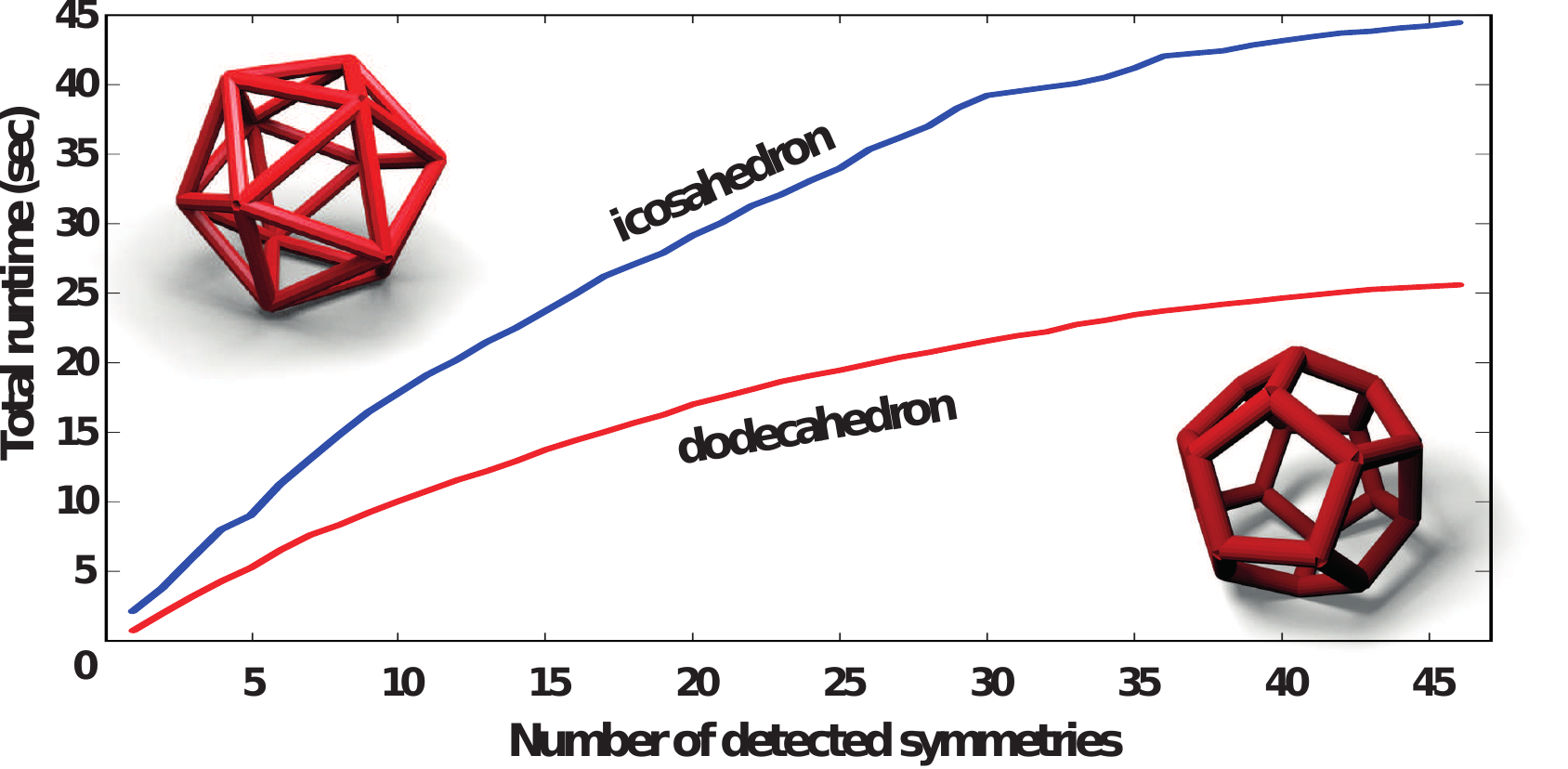} 
\caption{\textbf{Runtimes for the `icosahedron' and `dodecahedron'.}
The plot shows the runtime accumulating while detecting the symmetries.
The algorithm found exactly the 46 symmetries of these shapes: fifteen reflections, ten 3-fold rotations, six 5-fold rotations and fifteen 2-fold rotations.
Note that less effort is needed the further the algorithm progresses.
This is due to the fact that regions around the previously detected symmetries are being 'carved-out' from the search space. Our runtimes are competitive with those reported by Martinet {\em et al.}~\cite{martinet2006accurate}. See text for further details.}
\label{fig.runtimes2}
  \end{center}\vspace{-10pt}
\end{figure}

\begin{figure*}
\vspace{-3pt}
\begin{center}
\addtolength{\tabcolsep}{-0pt}
    \begin{tabular}{c c c c c }
      \includegraphics[height = 0.15\textheight]{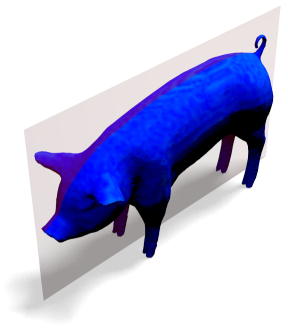} & 
      \includegraphics[height = 0.15\textheight]{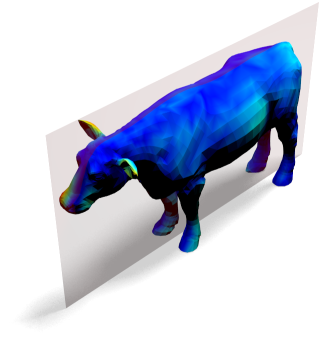} & 
      \includegraphics[height = 0.165\textheight]{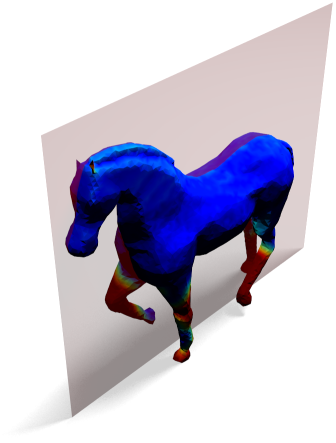} & 
      \includegraphics[height = 0.15\textheight]{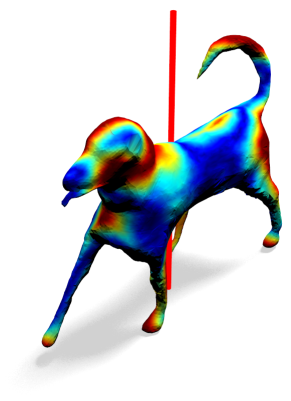} & 
      \includegraphics[height = 0.135\textheight]{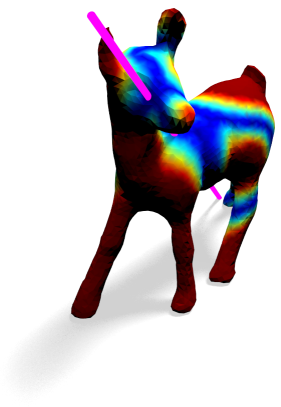} \vspace{-5pt}\\
      \small{REFL (0.0068)} & \small{REFL (0.0282)} & \small{REFL (0.0496)} & \small{2-ROT (0.1004)} & \small{CONT (0.1686)}\\
         \end{tabular}
        \vspace{-3pt}
        \caption{\textbf{Examples from the 'Fourlegs' category of COSEG \cite{wang2012active}, ordered by increasing distortion.} The single best detected symmetry is shown for each shape; symmetry type and distortion are reported below each image. The color of each point is the absolute difference between a point and its match in the transformed shape (increasing from blue to red). In the horse, for example, the body is symmetric but the legs are not.}
        \label{fig.COSEG.4leg}
    \end{center}
        \vspace{-20pt}
\end{figure*}

\subsection{Determining how symmetric a shape is}

Some applications only require to determine if a shape is symmetric or not. For example, in the case of quadrupeds, their shape will be symmetric if they are in the natural pose. To demonstrate this we took the `Fourlegs' class from the COSEG dataset \cite{wang2012active}, which includes 20 shapes of various quadrupeds. We ran our algorithm and found the best symmetry for each of the shapes. Figure~\ref{fig.COSEG.4leg} shows some of the shapes in increasing distortion order.

Several observations can be made. First, observe that the algorithm indeed found the most prominent symmetry in each case. The color code in Figure~\ref{fig.COSEG.4leg} encodes the difference between the corresponding points, measured as the absolute difference between a point on the original shape and its corresponding point on the transformed shape. As shown in the first example (pig), the shape is almost perfectly symmetric. In the next example (cow), the head is slightly tilted and indeed this is correctly detected by our algorithm. Note that this distortion does not affect the quality of the recovered plane of symmetry. The third example (horse) shows that the algorithm properly detects the symmetry plane despite the fact that the legs are not symmetric. The last two examples show shapes that are not symmetric and indeed the animals are not in their natural pose and the distortion of the best symmetry found by the algorithm is high.

\subsection{Symmetry in an MRI scalar volume}
So far we assumed that the input originates from solid 3D shapes. However, our method can handle general scalar volumes, which are common in 3D medical imaging. We used a simulated volume of a normal brain from~\cite{Cocosco97brainweb:online}, in the T1 MRI modality, slice thickness of $1$mm, $3\%$ noise and $20\%$ non-uniformity. The best symmetry detected by our method discovered the bilateral symmetry of the left and right brain hemispheres (see Figure \ref{fig.MRI}), although the model is far from being perfectly symmetric.

\section{Summary}

We presented a fast algorithm for global approximate 3D symmetry detection that is guaranteed to find all approximate symmetries of a volumetric representation of a shape within a user specified accuracy.
The algorithm is robust to noise and is fast in practice, taking about two seconds to detect a symmetry.

A key contribution of our work is a proof that the density of the net depends on the total variation of the shape.
Therefore, the best transformation on the net is within an approximation constant from the optimal transformation.
We further show the use of TSDF representations to control the shape total variation, and hence the sampling density.
The algorithm is further accelerated using sub-linear sampling that randomly examines only a small number of points, which makes the algorithm find symmetry with overwhelmingly high probability.

Several experiments asses the performance of the algorithm, including very complicated shapes with tens of symmetries.
Unlike previous work, we include an experiment on a large set of shapes and show that the method scales well.

The proposed algorithm can be modified and generalized in the following manners. First, it could handle richer transformation groups like Euclidean transformations $\mathrm{E}(3)$, enabling rigid registration of shapes. Second, it could be applied to a given part of a shape (as was done in \cite{FastMatchPaper}), and act as a component in a partial symmetry detector.

\bibliographystyle{plain}

\end{document}